\documentclass{article}
\usepackage{hyperref}

\usepackage{amssymb}
\usepackage{latexsym}
\usepackage{graphics}
\usepackage{amsmath}
\usepackage[thmmarks,amsmath,amsthm,hyperref]{ntheorem}

\usepackage{tikz}

\def\hhh{\vphantom{$\setminus,{}^{-1}$}}

\theoremstyle{plain} 
\newtheorem{theorem}{Theorem}[section]
\newtheorem{lemma}[theorem]{Lemma}
\newtheorem{proposition}[theorem]{Proposition}

\theoremstyle{definition}
\newtheorem{definition}[theorem]{Definition}

\theoremstyle{remark}
\newtheorem{remark}[theorem]{Remark}
\newtheorem{example}[theorem]{Example}

\newcommand{\BL}{\ensuremath{\mathcal{N}}}

\newcommand{\BLS}{\ensuremath{\BL_\Lambda}}

\newcommand{\fo}{\ensuremath{\mathrm{FO}}}

\newcommand{\fbar}{\ensuremath{\overline F}}
\newcommand{\ftilde}{\ensuremath{\widetilde{F}}}
\newcommand{\clbot}{\ensuremath{\mathcal{C}}}
\newcommand{\clcap}{\ensuremath{\mathcal{C}[\cap]}}

\newcommand{\alang}[1]{\ensuremath \overline{#1}}
\newcommand{\elang}[1]{\ensuremath \BL(#1)}
\newcommand{\zlang}[1]{\ensuremath\widehat{#1}}

\newcommand{\rel}{\ensuremath{R}}
\newcommand{\ler}{\ensuremath{\rel^{-1}}}
\newcommand{\id}{\ensuremath{\mathit{id}}}
\renewcommand{\div}{\ensuremath{\mathit{di}}}
\newcommand{\di}{\div}

\newcommand{\cpi}{\ensuremath{\overline\pi}}

\DeclareMathOperator{\adom}{adom}

\newcommand{\ltpath}{\leq^{\mathrm{path}}} 
\newcommand{\pathq}{\ltpath}
\newcommand{\ltbool}{\leq^{\mathrm{bool}}} 
\newcommand{\boolq}{\ltbool}

\newcommand{\nltstrong}{\not\ltbool_{\mathrm{strong}}}
\newcommand{\nltstrongpath}{\not\ltpath_{\mathrm{strong}}}

\newcommand{\mG}{\ensuremath{\mathbf{G}}} 

\newcommand{\lab}{\ensuremath{\Lambda}}

\newcommand{\conv}[1]{\ensuremath{{#1}^{-1}}}%

\newcommand{\ignore}[1]{}

\def\clap#1{\hbox to 0pt{\hss#1\hss}}

\tikzstyle{vertex}=[circle,fill=black,minimum size=4pt,inner sep=0pt]

\usetikzlibrary{decorations.pathmorphing}
\usetikzlibrary{decorations.markings}
\usetikzlibrary{arrows,calc,shapes}

\newcommand{\ttrue}{\textit{true}}
\newcommand{\ffalse}{\textit{false}}
\renewcommand{\setminus}{-}
\makeatletter
\gdef\SetFigFont#1#2#3#4#5{%
\reset@font\fontsize{#1}{#2pt}%
\fontfamily{#3}\fontseries{#4}\fontshape{#5}%
\selectfont}%
\makeatother

\begin{document}

\title{Relative Expressive Power of \\ Navigational Querying on
Graphs\thanks{An extended abstract announcing the results of
this paper was presented at the 14th International Conference on
Database Theory, Uppsala, Sweden, March 2011.}}
\author{George H.L. Fletcher \and Marc Gyssens \and Dirk Leinders
\and Dimitri Surinx \and Jan Van den Bussche \and Dirk Van Gucht
\and Stijn Vansummeren \and Yuqing Wu}

\maketitle

\begin{abstract}
  Motivated by both established and new applications, we study
  navigational query languages for graphs (binary relations).  The
  simplest language has only the two operators union and composition,
  together with the identity relation.  We make more powerful
  languages by adding any of the following operators: intersection;
  set difference; projection; coprojection; converse; and the
  diversity relation.  All these operators map binary relations to
  binary relations.  We compare the expressive power of all resulting
  languages.  We do this not only for general path queries (queries
  where the result may be any binary relation) but also for boolean or
  yes/no queries (expressed by the nonemptiness of an expression).
  For both cases, we present the complete Hasse diagram of relative
  expressiveness.  In particular the Hasse diagram for boolean
  queries contains some nontrivial separations and a few surprising
  collapses.
\end{abstract}

\section{Introduction}
\label{sec-introduction}
Graph databases, and the design and analysis of query languages
appropriate for graph data, have a rich history in database
systems and theory research \cite{gutierrez_survey}.  Originally
investigated from the perspective of object-oriented databases,
interest in graph databases research has been continually
renewed, motivated by data on the Web \cite{abs_book,flm_dbwebsurvey}
and new applications such as dataspaces \cite{halevy_dataspaces},
Linked Data \cite{bhb-l_linkedata}, and RDF \cite{rdfprimer}.

Typical of access to graph-structured data is its navigational
nature.  Indeed, in restriction to trees, there is a standard
navigational query language, called XPath, whose expressive power
has been intensively studied
\cite{benediktfankuper_xpath,marx_conditionalxpath}.
XPath has been formalized in terms of a number of basic operators on
binary relations \cite{marxrijke_xpath}.  Hence a natural
approach \cite{nsparql,gxpath,angles_pdl} is to take this same
set of operators but now evaluate them over graphs instead of
over trees.  Our goal in this paper is to understand the relative
importance of the different operators in this setting.

Concretely, in the present paper, we consider a number of natural
operators on binary relations (graphs): union; composition;
intersection; set difference; projection; coprojection; converse;
and the identity and diversity relations.  While some of these
operators also appear in XPath, they are there evaluated on
trees.  The largest language that we consider has all operators,
while the smallest language has only union, composition, and the
identity relation.  When a language has set difference, it also
has intersection, by $R \cap S = R - (R-S)$.  Interestingly, the
ensemble of all operators except intersection and set difference
precisely characterizes the first-order queries safe for
bisimulation \cite{benthem_safe,marxrijke_xpath}.  This logical
grouping of operators is also present in our research, where we
often have to treat the case without intersection separately from
the case with intersection.\footnote{Strictly speaking, van
Benthem's discussion \cite{benthem_safe} does not include the
converse operator nor the identity and diversity relations.}

Just as in the relational algebra,
expressions are built up from input relation names using these
operators.  Since each operator maps binary relations to binary
relations, these query languages express queries from binary
relations to binary relations: we call such queries \emph{path
queries}.  By identifying nonemptiness with the boolean value
`true' and emptiness with `false', as is standard in database
theory \cite{ahv_book}, we can also express yes/no queries
within this framework. To distinguish them from general path
queries, we shall refer to the latter as \emph{boolean queries}.

The contribution of the present paper is providing a complete
comparison of the expressiveness of all resulting languages, and this
both for general path queries and boolean queries.  While establishing
the relative expressiveness for general path queries did not yield
particularly surprising results, the task for the case of boolean
queries proved much more challenging.  For example, consider the
converse operator $R^{-1} = \{(y,x) \mid (x,y) \in R\}$.  On the one
hand, adding converse to a language not yet containing this feature
sometimes adds boolean query power. This is, e.g., the case for the
language containing all other features.  The proof, however, is
nontrivial and involves a specialized application of invariance under
bisimulation known from arrow logics. On the other hand, adding
converse to a language containing projection but not containing
intersection does not add any boolean query power.  We thus obtain a
result mirroring similar results known for XPath on trees
\cite{benediktfankuper_xpath,olteanu_forward,wu_pospathfragment},
where, e.g., downward XPath is known
to be as powerful as full XPath for queries evaluated at the root.

Let us briefly discuss some of the
methods we use.  In many cases where we separate a language
$\mathcal{L}_1$ from a language $\mathcal{L}_2$, we can do this in a
strong sense: we are able to give a single counterexample, consisting
of a pair $(A,B)$ of finite binary relations such that $A$ and $B$ are
distinguishable by an expression from $\mathcal{L}_1$ but
indistinguishable by any expression from $\mathcal{L}_2$.  Notice that
in general, separation is established by providing an infinite
sequence of relation pairs such that some expression from
$\mathcal{L}_1$ distinguishes all pairs but no expression of
$\mathcal{L}_2$ distinguishes all pairs. Existence of a single
counterexample pair is therefore nonobvious, and we do not really know
whether there is a deeper reason why in our setting this strong form
of separation can often be established. Strong separation is desirable
as it immediately implies separation of $\mathcal{L}_1$ not only from $\mathcal{L}_2$ but also from the infinitary variant of $\mathcal{L}_2$ (which allows infinite unions, as in infinitary logic \cite{ef_fmt2}). Note that indistinguishability of a pair of finite binary relations can in
principle be checked by computer, as the number of possible binary
relations on a finite domain is finite.
Indeed, in many cases we have used this ``brute-force approach'' to
verify indistinguishability.
In some cases, however, this
approach is not feasible within a reasonable time.
Fortunately, by applying invariance under bisimulation for arrow
logics \cite{marxvenema_multi},
we can alternatively check a sufficient condition
for indistinguishability in polynomial time.  We have applied this
alternative approach in our computer checks. Finally, the cases where
we could not establish strong separation fall in the class of
conjunctive queries \cite{ahv_book}.
We developed a method based on homomorphism
techniques to establish ordinary separation for these cases.

The languages considered here are very natural and date all the
way back to the ``calculus of relations'' created by Peirce and
Schr\"oder, and popularized and greatly developed by Tarski and
his collaborators \cite{tarski_relcalc,tarskigivant}.  The full
language actually has the same expressive power as 3-variable
first-order logic (FO$^3$) under the active-domain semantics, for
path queries as well as for boolean queries.  Due to the
naturalness of the languages, they appear in many other fields
where binary relations are important, such as description logics,
dynamic logics, arrow logics, and relation algebras
\cite{dlhandbook,dynamiclogicbook,marxvenema_multi,handbookmodallogic,maddux_book,hh_relalgames}.
Thus, our results also yield some new insight into these fields.
The investigation of expressive power as in the present paper is
very natural from a database theory perspective.  In the
above-mentioned fields, however, one is primarily interested in
other questions, such as computational complexity of model
checking, decidability of satisfiability, and axiomatizability of
equivalence.  The expressiveness issues investigated in this
paper have not been investigated before.\footnote{Strictly
speaking, one may argue that the ``calculus of relations ''
refers to a set of equational axioms now known as the axioms for
relation algebras (see the references above).  However, the
original and natural interpretation of the operations of the
calculus of relations is clearly that of operations on binary
relations \cite{tarskigivant,pratt_relcalc}.  In modern
terminology this interpretation corresponds to `representable'
relation algebras.  We stress that the present paper
focuses on the expressive power of the various operations and not on
axiomatizability, completeness of equations, or representability
of abstract relation algebras}.

At this point we must repeat that
also in the database field, graph query languages have been
investigated intensively.  There is, for example, the vast body
of work on conjunctive regular path queries (CRPQs)
\cite{barcelo_crpq_pods}.  As a matter of fact, CRPQs are
subsumed in the calculus of relations, with the exception of the
Kleene star (transitive closure) operator.  Indeed, the results
reported in this journal article have been extended to the setting where
transitive closure is present, as originally announced in our
conference paper \cite{rafragments}.  This extension will be
elaborated in a companion journal article \cite{nav_with_tc};
additional results on the special case of a single relation name
have been published in a third journal article \cite{amai_dipi}.

This paper is further organized as follows. In Section\nobreakspace \ref {sec-prelim}, we
define the class of languages studied in the paper. In
Section\nobreakspace \ref {sec:resol-expr}, we describe the techniques we use to
separate one language from another.  In section~\ref{sec:technical-results} we present our two main technical results in a self-contained manner: first, the added power of projection in expressing boolean queries, compared to the language without intersection and coprojection; second, the elimination of converse in languages with projection, but without intersection.
Then we establish the
complete Hasse diagram of relative expressiveness. We do so for
path queries in Section\nobreakspace \ref {sec:path-queries}, and for boolean queries
in Section\nobreakspace \ref {sec:boolean-queries}.  Finally, we discuss future
research directions in Section\nobreakspace \ref {sec:further-research}.

\section{Preliminaries}
\label{sec-prelim}

In this paper, we are interested in navigating over graphs whose
edges are labeled by symbols from a finite, nonempty set of
labels $\lab$.  We can regard these edge labels as binary
relation names and thus regard $\lab$ as a relational database
schema.  For our purposes, then, a \emph{graph} $G$ is an instance of
this database schema $\lab$.  That is, assuming an infinite
universe $V$ of data elements called \emph{nodes}, $G$ assigns to
every $R \in \lab$ a relation $G(R) \subseteq V\times V$.  Each
pair in $G(R)$ is called an \emph{edge} with label $R$.
In what follows, $G(R)$
may be infinite, unless explicitly stated otherwise.  All
inexpressibility results in this paper already hold in
restriction to finite graphs, however. 

The most basic language for navigating over graphs we consider is the
algebra $\BL$ whose expressions are built recursively from the edge
labels, the primitive $\emptyset$, and the primitive $\id$, using
composition ($e_1 \circ e_2$) and union ($e_1 \cup
e_2$). Semantically, each expression $e \in \BL$ defines a path
query. A \emph{path query} is a function $q$ taking any graph $G$ as
input and returning a binary
relation $q(G) \subseteq \adom(G) \times \adom(G)$. Here, $\adom(G)$
denotes the \emph{active domain} of $G$, which is the set of all
entries occurring in one of the relations of $G$.  Formally,
\[ \adom(G) = \{m \mid \exists n, \exists R \in \lab : (m,n) \in G(R) \lor
(n,m) \in G(R)\}.\]

In detail, the semantics of $\BL$ is inductively defined as follows:
\begin{align*}
\rel(G)&=G(\rel)\,{\rm;}\\
\emptyset(G)& = \emptyset\,\rm;\\
\id(G)&=\{(m,m)\mid m \in \adom(G)\}\,\rm;\\
e_1\circ e_2(G)&=\{(m,n)\mid \exists p\, ((m,p)\in e_1(G)\ \&\ (p,n)\in
e_2(G))\}\,\rm; \\
e_1 \cup e_2(G) & =  e_1(G) \cup e_2(G)\,\rm.
\end{align*}
\noindent The basic algebra $\BL$ can be extended by adding some of
the following features: diversity ($\div$), converse ($\conv e$),
intersection ($e_1\cap e_2)$, difference ($e_1\setminus e_2$),
projections ($\pi_1(e)$ and $\pi_2(e)$), and the coprojections
($\cpi_1(e)$ and $\cpi_2(e)$). We refer to the operators in the basic
algebra $\BL$ as \emph{basic features}; we refer to the extensions as
\emph{nonbasic features}. The semantics of the extensions is as
follows:
\begin{align*}
\div(G)&=\{(m,n)\mid m,n\in\adom(G)\ \&\ m\not=n\}\,\rm;\\
\conv e(G)&=\{(m,n)\mid (n,m)\in e(G)\}\,\rm;\\
e_1 \cap e_2(G) & =  e_1(G) \cap e_2(G)\,\rm;\\
e_1 \setminus e_2(G) & =  e_1(G) \setminus e_2(G)\,\rm;\\
\pi_1(e)(G)&=\{(m,m)\mid m\in\adom(G)\ \&\ \exists n\, (m,n)\in e(G)\}\,\rm;\\
\pi_2(e)(G)&=\{(m,m)\mid m\in\adom(G)\ \&\ \exists n\, (n,m)\in e(G)\}\,\rm;\\
\cpi_1(e)(G)&=\{(m,m)\mid m\in\adom(G)\ \&\ \lnot\exists n\, (m,n)\in
e(G)\}\,\rm;\\
\cpi_2(e)(G)&=\{(m,m)\mid m\in\adom(G)\ \&\ \lnot\exists n\, (n,m)\in
e(G)\}\,\rm.
\end{align*}

If $F$ is a set of nonbasic features, we denote by $\BL(F)$ the
language obtained by adding all features in $F$ to $\BL$. For example,
$\BL(\cap)$ denotes the extension of $\BL$ with intersection, and
$\BL(\cap, \pi)$ denotes the extension of $\BL$ with intersection and
both projections.\footnote{We do not consider extensions of $\BL$ in
  which only one of the two projections, respectively one of the two
  coprojections, is present.}
We will see below that extending the basic algebra with
diversity, difference, and converse is sufficient to express all
other nonbasic features.  This full language
$\BL(\setminus,\div,\conv{})$ is known
as the \emph{calculus of relations}.

We will actually compare language expressiveness at the level of both
path queries and boolean queries.  Path queries were defined
above; a \emph{boolean query} is a function from graphs to
$\{\mathrm{true},\mathrm{false}\}$.

\begin{definition}
  A path query $q$ is expressible in a language $\BL(F)$ if there
  exists an expression $e \in \BL(F)$ such that, for every graph $G$,
  we have $e(G) = q(G)$. Similarly, a boolean query $q$ is expressible
  in $\BL(F)$ if there exists an expression $e \in \BL(F)$ such that,
  for every graph $G$, we have that $e(G)$ is nonempty if, and only
  if, $q(G)$ is true. In both cases, we say that \emph{$q$ is
    expressed by $e$}.
\end{definition}

In what follows, we write $\BL(F_1) \ltpath \BL(F_2)$ if every path
query expressible in $\BL(F_1)$ is also expressible in
$\BL(F_2)$. Similarly, we write $\BL(F_1) \ltbool \BL(F_2)$ if every
boolean query expressible in $\BL(F_1)$ is also expressible in
$\BL(F_2)$. Note that $\BL(F_1) \ltpath \BL(F_2)$ implies $\BL(F_1)
\ltbool \BL(F_2)$, but not necessarily the other way around. We write
$\not\ltpath$ and $\not\ltbool$ for the negation of $\ltpath$ and
$\ltbool$.
\begin{remark}
The attentive reader will note that every fragment $\elang{F}$ actually depends on the label vocabulary $\Lambda$ which is arbitrary but fixed. So to be fully precise we would need to use the notation $\BLS(F)$. For all the results in this paper, a comparison of fragments of the form $\elang{F_1} \leq \elang{F_2}$ (with $\leq$ being $\ltpath$ or $\ltbool$) can be interpreted to mean that we have $\BLS(F_1) \leq \BLS(F_2)$ for every $\Lambda$. Moreover, whenever we have a negative result of the form $\elang{F_1} \not\leq \elang{F_2}$, this will actually already hold for the simplest $\Lambda$ consisting of a single label. 

To illustrate, in the interpretation described above, the $\id$ relation may be considered redundant in any fragment that includes the projections. Indeed, we can express $\id$ as $\bigcup_{R\in \Lambda} (\pi_1(R) \cup \pi_2(R))$. This observation falls outside the scope of the present investigation, however, since we do not consider $\id$ as an optional feature; it belongs to all fragments considered in this paper.
\end{remark}

\begin{remark}
The language XPath \cite{xpath} also includes the path equality operator
$.[e_1=e_2]$ (in XPath called `general comparison'),
with the following semantics:
$$ .[e_1=e_2](G) = \{(m,m)\mid m\in\adom(G)\ \&\ \exists n\,
(m,n) \in e_1(G) \cap e_2(G)\}. $$
This operator can be expressed in the fragment $\mathcal
N(\pi,\cap)$ as $\pi_1(e_1 \cap e_2)$, as well as in the fragment $\mathcal
N(\conv{},\cap)$ as $(e_1 \cap e_2^{-1}) \cap \id$.  Actually the
latter expression is not particular to this example, because it
reflects the way in which projection is expressed using converse
and intersection, as we will see in
Section~\ref{sec:path-queries}.
\end{remark}

\section{Tools to establish separation}
\label{sec:resol-expr}

Our results in Section\nobreakspace \ref {sec:path-queries} and
\ref{sec:boolean-queries} will use the following tools to separate a
language $\BL(F_1)$ from a language $\BL(F_2)$, i.e., to
establish that
$\BL(F_1) \not\ltpath \BL(F_2)$, or $\BL(F_1) \not\ltbool \BL(F_2)$. It will also be useful to consider stronger
variants of $\not\ltpath$ and $\not\ltbool$.
\begin{definition} \label{def:strong} The language $\BL(F_1)$ is
  \emph{strongly separable from} the language $\BL(F_2)$ \emph{at the
    level of path queries} if there exists a path query $q$
  expressible in $\BL(F_1)$ and a finite graph $G$, such that, for
  every expression $e \in \BL(F_2)$, we have $q(G) \neq e(G)$. We
  write $\BL(F_1) \nltstrongpath \BL(F_2)$ in this case. Similarly,
  $\BL(F_1)$ is \emph{strongly separable from} $\BL(F_2)$ \emph{at the
    level of boolean queries} if there exists a boolean query $q$
  expressible in $\BL(F_1)$ and two finite graphs $G_1$ and $G_2$,
  with $q(G_1)$ true and $q(G_2)$ false, such that, for every
  expression $e \in \BL(F_2)$, $e(G_1)$ and $e(G_2)$ are both empty,
  or both nonempty. We write $\BL(F_1) \nltstrong \BL(F_2)$ in this
  case.
\end{definition}

\subsection{Path separation}
\label{sec:strong-path-separation}

Since $\BL(F_1) \ltpath \BL(F_2)$ implies $\BL(F_1) \ltbool \BL(F_2)$,
also $\BL(F_1) \not\ltbool \BL(F_2)$ implies $\BL(F_1) \not\ltpath
\BL(F_2)$ by contraposition. In most instances, we can therefore
establish separation at the level of general path queries by
establishing separation at the level of boolean queries. In the cases
where $\BL(F_1) \not\ltpath \BL(F_2)$ although $\BL(F_1) \ltbool
\BL(F_2)$, we identify a finite graph $G$ and an expression $e_1$ in
$\BL(F_1)$ and show that, for each expression $e_2$ in $\BL(F_2)$,
$e_1(G) \neq e_2(G)$. Notice that we actually establish strong path
separation in those cases.

\subsection{Boolean separation}
\label{sec:boolean-separation}

To establish separation at the level of boolean queries, we use the
following techniques.

\subsubsection{Brute-force approach}
\label{sec:brute-force-approach}

Two graphs $G_1$ and $G_2$ are said to be \emph{distinguishable} at the
boolean level in a language $\BL(F)$ if there exists a boolean query
$q$ expressible in $\BL(F)$ such that exactly one of $q(G_1)$ and
$q(G_2)$ is true, and the other is false. If such a query does not
exists, $G_1$ and $G_2$ are said to be \emph{indistinguishable} in $\BL(F)$.

Using this terminology, two languages $\BL(F_1)$ and $\BL(F_2)$ are
\emph{strongly separable} if there exist two finite graphs $G_1$ and $G_2$
that are distinguishable in $\BL(F_1)$, but indistinguishable in
$\BL(F_2)$.

For two finite graphs $G_1$ and $G_2$, (in)distinguishability in a
language $\BL(F)$ can easily be machine-checked through the
Brute-Force Algorithm described below.

First observe that $\adom(G_1)$ and $\adom(G_2)$ are finite since
$G_1$ and $G_2$ are finite. Moreover, for any $e$ in $\BL(F)$, $e(G_1)
\subseteq \adom(G_1) \times \adom(G_1)$ and $e(G_2) \subseteq
\adom(G_2) \times \adom(G_2)$. Hence, $e(G_1)$ and $e(G_2)$ are finite
and the set $\{(e(G_1), e(G_2)) \mid e \in \BL(F) \}$ is also
finite. Clearly, $G_1$ is indistinguishable from $G_2$ if this set
contains only pairs that are both empty or both nonempty.

The Brute-Force Algorithm computes the above set by first initializing 
the set
\[
B = \{ (\id(G_1), \id(G_2))\}\, \cup\, \{ (\div(G_1), \div(G_2))\}\, \cup\, \{
(G_1(R), G_2(R)) \mid R \in \lab \}
\]
(where $\{ (\div(G_1), \div(G_2))\}$ is omitted if $\div \not \in
F$). It then adds new pairs $(R_1, R_2)$ to $B$ by closing $B$
pair-wise under the features in $\BL(F)$. That is, for every binary
operator $\otimes$ in $\BL(F)$ and all pairs $(R_1,R_2), (S_1,S_2)$ in
$B$ the algorithm adds $(R_1 \otimes S_1, R_2 \otimes S_2)$ to $B$,
and similarly for the unary operators. Since there are only a finite
number of pairs, the algorithm is guaranteed to end. Of course, the
worst-case complexity of this brute-force algorithm is
exponential. Nevertheless, we have successfully checked
indistinguishability using this Brute-Force Algorithm in many of the
cases that follow.

\subsubsection{Bisimulation}
\label{sec:bisim-invar}

We will not always be able to use the methodology above to separate
two languages.  In particular, to establish that $\BL(\conv{},\cap)
\not \ltbool \allowbreak \BL(\setminus,\div)$ we will employ
invariance results under the notion of bisimulation below.  In
essence, this notion is based on the notion of bisimulation known from
arrow logics \cite{marxvenema_multi}.
Below, we adapt this notion to the current setting.

We require the following preliminary definitions.  Let $\mG=(G,a,b)$
denote a \emph{marked graph}, i.e., a graph $G$ with $a,b\in\adom(G)$.
The \emph{degree} of an expression $e$ is the maximum depth of nested
applications of composition, projection and coprojection in $e$.  For
example, the degree of $R \circ R$ is 1, while the degree of both $R
\circ (R \circ R)$ and $\pi_1(R \circ R)$ is 2.  Intuitively, the
depth of $e$ corresponds to the quantifier rank of the standard
translation of $e$ into $\fo^3$. For a set of features $F$, $\BL(F)_k$
denotes the set of expressions in $\BL(F)$ of degree at most~$k$.

In what follows, we are only concerned with bisimulation results
regarding $\BL(\setminus, \div)$. The following is an appropriate
notion of bisimulation for this language.

\begin{definition}[Bisimilarity]
\label{def-diff-bisimilar}

Let $k$ be a natural number, and let $\mG_1=(G_1,a_1,b_1)$ and
$\mG_2 = (G_2,a_2, b_2)$ be marked graphs. We say that
$\mG_1$ is bisimilar to $\mG_2$ up to depth
$k$, denoted $\mG_1\simeq_k \mG_2$, if the
following conditions are satisfied:
\begin{description}

\item[Atoms] $a_1 = b_1$ if and only if $a_2 = b_2$; and $(a_1,b_1)
  \in G_1(R)$ if and only if $(a_2,b_2) \in G_2(R)$, for every $R \in
  \lab$;

\item[Forth] if $k>0$, then, for every $c_1$ in $\adom(G_1)$, there exists
    some $c_2$ in $\adom(G_2)$ such that
    both $(G_1,a_1,c_1)\simeq_{k-1}(G_2,a_2,c_2)$ and
    $(G_1,c_1,b_1)\ \simeq_{k-1} (G_2,c_2,b_2)$;

\item[Back] if $k>0$, then, for every $c_2$ in $\adom(G_2)$, there
    exists some $c_1$ in $\adom(G_1)$ such that
    both $(G_1,a_1,c_1)\simeq_{k-1}(G_2,a_2,c_2)$ and
    $(G_1,c_1,b_1)\ \simeq_{k-1} (G_2,c_2,b_2)$.
\end{description}
\end{definition}

Expressions in $\BL(\setminus,\div)$ of depth at most $k$ are
invariant under bisimulation:

\begin{proposition}
\label{theo-diff-bisimilar}
Let $k$ be a natural number; let $e$ be an expression in
$\BL(\setminus,\div)_k$; and let $\mG_1 = (G_1,a_1,b_1)$ and $\mG_2 =
(G_2,a_2,b_2)$ be marked graphs. If $\mG_1\simeq_k \mG_2$ then
$(a_1,b_1) \in e(G_1) \Leftrightarrow (a_2,b_2) \in e(G_2)$.
\end{proposition}
In other words, if $\mG_1\simeq_k \mG_2$, then any expression of
degree at most $k$ either both selects $(a_1,b_1)$ in $G_1$ and
$(a_2,b_2)$ in $G_2$, or neither of them. As such, the marked graphs
$\mG_1$ and $\mG_2$ are \emph{indistinguishable} by expressions in
$\BL(\setminus,\div)_k$.  The proof of
Proposition\nobreakspace \ref {theo-diff-bisimilar} is by a straightforward
induction on $e$.

The following proposition states how we can use
Proposition\nobreakspace \ref {theo-diff-bisimilar} to show that some boolean query
is not expressible in $\BL(\setminus, \div)_k$.

\begin{proposition}
\label{prop-bisimilar-noparameter}
Let $k$ be a natural number.  A boolean query $q$ is not expressible
in ${\BL(\setminus,\div)}_k$ if there exist graphs $G_1$ and $G_2$
such that $q(G_1)$ is true and $q(G_2)$ is false, and, for each
pair $(a_1,b_1) \in \adom(G_1)^2$, there exists $(a_2,b_2) \in
\adom(G_2)^2$ such that $(G_1,a_1,b_1)\simeq_k(G_2,a_2,b_2)$.
\end{proposition}

We omit the straightforward proof; we note that the converse
implication holds as well \cite{rabisim}.

\subsubsection{Homomorphism approach}
\label{sec:homom-approach}

To show that $\BL(\pi) \not\ltbool \BL(\conv{}, \div)$, we used
an entirely different technique, based on the theory of
conjunctive queries and the nonexistence of certain homomorphisms on
particular graphs.  The details are given in Section~\ref{sec:zigzag}.
\section{The power of various operators}\label{sec:technical-results}
In this section, two main technical results are shown regarding the power of various operators. 
The first result (Proposition~\ref{bottom-pi-tech}) states that the $\pi$ operator (in combination with the basic operators) provides some boolean querying power that cannot be provided by the ${}^{-1}$ and $\di$ operators. This is a sharp expressivity result on projection, since adding any other feature to the fragment $\elang{{}^{-1},\di}$ leads to the expressibility of projection. 
\begin{proposition}
\label{bottom-pi-tech} $\elang{\pi}\not\ltbool \elang{{}^{-1},\di}$.
\end{proposition}
Since this result is highly technical, it is proven in Section~\ref{sec:zigzag}.

The second result (Proposition\nobreakspace \ref {prop:converse-elimination}) shows that, at the level
of boolean queries, $\conv{}$ does not add expressive power in the
presence of $\pi$ and in the absence of $\cap$.  

\begin{proposition} \label{prop:converse-elimination} Let $F$ be a set
  of nonbasic features for which $\setminus \not \in F$ and $\cap \not\in F$. Then, $\BL(F
  \cup \{ \conv{} \}) \ltbool \BL(F \cup \{ \pi \})$.
\end{proposition}

\begin{example}\label{ex:converse-elimination}
  To illustrate Proposition\nobreakspace \ref {prop:converse-elimination}, consider
  the expression $e_1= R^3 \circ \conv{R} \circ R^3$ in
  $\BL(\conv{})$. The expression $\pi_1(e_1)$ can be equivalently
  expressed in $\BL(\pi)$ as $\pi_1\big(R^3 \circ \pi_2( \pi_1(R^3)
  \circ R )\big)$. Now observe that, for any graph $G$, we have that
  $e_1(G)$ is nonempty if and only if $\pi_1(e_1)(G)$ is nonempty.

  Using this same observation, one can express the non-emptiness of
  the expression $e_2 = R \circ \cpi_2( (R\circ S) \cup (\conv{R}
  \circ S))$ in $\BL(\conv{}, \cpi)$ by the non-emptiness of the
  expression $ \pi_1(e_2) = \pi_1\big( R \circ \cpi_2(R\circ S) \circ
  \cpi_2( \pi_1(R) \circ S ) \big)$ in $\BL(\cpi)$ .
\end{example}

\begin{proof}[Proof of Proposition\nobreakspace \ref {prop:converse-elimination}]

Let $e$ be an expression in $\BL(F \cup \{ \conv{}, \pi \})$.  Without
loss of generality, we may assume that $\conv{}$ is only applied in
$e$ to edge labels, so for each edge label $R$ we also consider
$\conv{R}$ as an edge label. By simultaneous induction on the size of
$e$ (the number of nodes in the syntax tree), we prove
for $i=1,2$ that
\begin{itemize}
\item
$\pi_i(e)$ is expressible in $\BL(F \cup \{ \pi \})$; and
\item 
if $\cpi \in \fbar$, then $\cpi_i(e)$ is expressible in
  $\BL(F)$.
\end{itemize}
Notice that the second statement is implied by the first, but we need
to consider both statements together to make the induction work. The basis of the induction is trivial. For all operators except composition we reason as follows:
\begin{align*}
\pi_1(\conv R) & = \pi_2(R) & \cpi_1(\conv R) & = \cpi_2(R) \\
\pi_2(\conv R) & = \pi_1(R) & \cpi_2(\conv R) & = \cpi_1(R) \\
\pi_i(\pi_j(e')) & = \pi_j(e') & \cpi_i(\pi_j(e')) & = \cpi_j(e') \\
\pi_i(\cpi_j(e')) & = \cpi_j(e') & \cpi_i(\cpi_j(e')) & = \pi_j(e') \\
\pi_i(e_1 \cup e_2) & = \pi_i(e_1) \cup \pi_i(e_2) &
\cpi_i(e_1 \cup e_2) & = \cpi_i(e_1) \circ \cpi_i(e_2).
\end{align*}

This leaves the case where $e$ is of the form $e_1 \circ e_2$.
Let $n$ be the first node in preorder in the syntax tree of
$e$ that is not an application of $\circ$, and let
$e_3$ be the expression rooted at $n$.  By associativity of
$\circ$, we can equivalently write $e$ in the form $e_3 \circ
e_4$, where $e_4$ equals the composition of all right-child
expressions from the parent of $n$ up to the root (in that
order).  Note that $e_3 \circ e_4$ has the same size as $e$.
We now consider the different possibilities for the form of $e_3$:
\begin{align*}
\pi_1(\id \circ e_4) & = \pi_1(e_4) \\
\pi_1(\div \circ e_4) & = \pi_1(\div \circ \pi_1(e_4)) \\
\pi_1(R \circ e_4) & = \pi_1(R \circ \pi_1(e_4)) \\
\pi_1(\conv R \circ e_4) & = \pi_2(\pi_1(e_4) \circ R) \\
\pi_1(\pi_j(e_5) \circ e_4) & =  \pi_j(e_5) \circ \pi_1(e_4) \\
\pi_1(\cpi_j(e_5) \circ e_4) & = \cpi_j(e_5) \circ \pi_1(e_4) \\
\pi_1((e_5 \cup e_6) \circ e_4) & = \pi_1(e_5 \circ e_4) \cup
\pi_1(e_6 \circ e_4) \\
\cpi_1(\id \circ e_4) & = \cpi_1(e_4) \\
\cpi_1(\div \circ e_4) & = \cpi_1(\div \circ \pi_1(e_4)) \\
\cpi_1(R \circ e_4) & = \cpi_1(R \circ \pi_1(e_4)) \\
\cpi_1(\conv R \circ e_4) & = \cpi_2(\pi_1(e_4) \circ R) \\
\cpi_1(\pi_j(e_5) \circ e_4) & =  \cpi_j(e_5) \cup \cpi_1(e_4) \\
\cpi_1(\cpi_j(e_5) \circ e_4) & = \pi_j(e_5) \cup \cpi_1(e_4) \\
\cpi_1((e_5 \cup e_6) \circ e_4) & = \cpi_1(e_5 \circ e_4) \circ
\cpi_1(e_6 \circ e_4)
\end{align*}

The crucial rules that eliminate inverse in the composition step are the fourth and the fourth-last. Hence we prove their correctness formally. Let $G$ be an arbitrary graph. Then, 
\begin{align*}(x,x) \in \pi_1(R^{-1}\circ e_4)(G) &\Leftrightarrow \exists y: (x,y) \in R^{-1}\circ e_4(G)\\
                                                  &\Leftrightarrow \exists y\exists z: (x,z) \in R^{-1}(G) \land (z,y) \in e_4(G)\\
                                                  &\Leftrightarrow \exists z: (z,x) \in R(G)\land (z,z)\in \pi_1(e_4)(G)\\
                                                  &\Leftrightarrow \exists z: (z,x)\in \pi_1(e_4)\circ R(G)\\
                                                  &\Leftrightarrow (x,x)\in \pi_2(\pi_1(e_4)\circ R))(G).
\end{align*} 
This proves the fourth rule. The fourth-last rule follows from the fourth rule and the fact that $\cpi_i(e') = \id\setminus \pi_i(e')$. This handles $\pi_1(e)$ and $\cpi_1(e)$.

To handle $\pi_2(e)$ and $\cpi_2(e)$, let $n$ now be
the first node in reverse preorder that is not an application of
$\circ$.  We can now write $e$ as $e_4 \circ e_3$.  The proof is
now similar:
\begin{align*}
\pi_2(e_4 \circ \id) & = \pi_2(e_4) \\
\pi_2(e_4 \circ \div) & = \pi_2(\pi_2(e_4) \circ \div) \\
\pi_2(e_4 \circ R) & = \pi_2(\pi_2(e_4) \circ R) \\
\pi_2(e_4 \circ \conv R) & = \pi_1(R \circ \pi_2(e_4)) \\
\pi_2(e_4 \circ \pi_j(e_5)) & = \pi_2(e_4) \circ \pi_j(e_5) \\
\pi_2(e_4 \circ \cpi_j(e_5)) & = \pi_2(e_4) \circ \cpi_j(e_5) \\
\pi_2(e_4 \circ (e_5 \cup e_6)) & = \pi_2(e_4 \circ e_5) \cup
\pi_2(e_4 \circ e_6) \\
\cpi_2(e_4 \circ \id) & = \cpi_2(e_4) \\
\cpi_2(e_4 \circ \div) &= \cpi_2(\pi_2(e_4) \circ \div) \\
\cpi_2(e_4 \circ R) &= \cpi_2(\pi_2(e_4) \circ R) \\
\cpi_2(e_4 \circ R^{-1}) &= \cpi_1(R \circ \pi_2(e_4)) \\
\cpi_2(e_4 \circ \pi_j(e_5)) &= \cpi_j(e_5) \cup \cpi_2(e_4) \\
\cpi_2(e_4 \circ \cpi_j(e_5)) &= \pi_j(e_5) \cup \cpi_2(e_4) \\
\cpi_2(e_4 \circ (e_5 \cup e_6)) &= \cpi_2(e_4 \circ e_5) \circ
\cpi_2(e_4 \circ e_6).
\end{align*}
In particular, if $e$ is an expression in $\BL(F \cup \{ \conv{} \})$,
it follows from the above that $\pi_1(e)$ is expressible in $\BL(F
\cup \{ \pi \})$. Proposition\nobreakspace \ref {prop:converse-elimination} now
follows from the observation that, for any graph $G$, $e(G)$ is
nonempty if and only if $\pi_1(e)(G)$ is nonempty.
\end{proof}
\begin{remark}
Proposition~\ref{prop:converse-elimination} may remind one of a similar result known for XPath on trees
\cite{benediktfankuper_xpath,olteanu_forward,wu_pospathfragment}
where downward XPath is known to be as  
powerful as full XPath for queries evaluated at the root. However, an important difference is that we are using projections both on the first and second column of a relation, whereas in the result on trees only the first projection is present. 

Indeed, Proposition~\ref{prop:converse-elimination} no longer holds for a language which only contains the first, but not the second projection, or vice versa. Consider the following two graphs $G_1 = \{ R(a,b), S(c,b) \}$ en $G_2 = \{ R(a,b), S(c,d) \}$. For any expression $e \in \elang{\pi_1}$ it must be that $e(G_1) \subseteq \{(a,a),(b,b),(c,c),(a,b),(c,b)\}$. It is not hard to see that for each $(x,y) \in \{(a,a),(b,b),(c,c),(a,b)\}$, $(x,y) \in e(G_1)$ iff $(x,y)\in e(G_2)$ and $(c,b) \in e(G_1)$ iff $(c,d)\in e(G_2)$. Therefore, it is clear that $G_1$ and $G_2$ are indistinguishable in $\elang{\pi_1}$. They are, however, distinguishable in $\elang{{}^{-1}}$ by $R\circ S^{-1}$.
\end{remark}
\begin{remark}
Notice that the translation used to eliminate converse in the proof of Proposition~\ref{prop:converse-elimination} could blow-up the size of the expressions exponentially. Indeed, define a family of expressions inductively as follows: $e_0 = T$ and $e_{n+1} = \pi_1((R\cup T) \circ e_n)$. Let us denote the size of an expression $e$ as $|e|$.  Clearly, $|e_0| = 0$ and $|e_{n+1}| = |e_n|+5$, which implies that $|e_n|$ is linear in $n$. Now, let $e_n'$ be the expression formed from $e_n$ according to the rules outlined in the proof of Proposition~\ref{prop:converse-elimination}. Clearly, $e_0' = T$ and $e_{n+1}' = \pi_1(R\circ e_n') \cup \pi_1(S\circ e_n')$. Therefore, $|e_0'| = 1$ and $|e_{n+1}'| = 2|e_{n}'| + 7$, which implies that $|e_n'| \geq 2^n$. 

On the other hand, our translation is never worse than single-exponential. We leave open whether a polynomial translation is possible. Interestingly, the analogous question about the complexity of translating from FO$^3$ to $\elang{\di,{}^{-1},\setminus}$, mentioned in the Introduction, has not yet been addressed in the literature. For fragments of FO$^2$, a relevant result has been reported~\cite{evw_fo2}.
\end{remark}
\subsection{Proof of Proposition~\protect\ref{bottom-pi-tech}\label{sec:zigzag}}

We begin by recalling some basic terminology and notions
concerning conjunctive queries \cite{ahv_book}. A
\emph{conjunctive query with nonequalities} is expressed in the
form $H\leftarrow B$. Here the body $B$ is a finite set of
relation atoms over the vocabulary $\Lambda$, as well as
nonequalities of the form $x\neq y$. The head $H$ is a tuple of
variables from $B$. The head may be the empty tuple in which case a boolean query is expressed.

Given a conjunctive query $Q$: $H\leftarrow B$ and a graph $G$,
an \emph{assignment} is a function $f$ from the set of variables
in $Q$ to $\adom{(G)}$. We call $f$ a \emph{matching} of $B$ in
$G$ if for each relation atom $R(x,y)$ in $B$, we have
$(f(x),f(y)) \in R(G)$, and for each $x\neq y$ in $B$ we have
$f(x) \neq f(y)$. The evaluation of $Q$ on $G$ is then defined as
\[ Q(G) = \{ f(H)\mid \text{$f$ is a matching from $B$ to $G$}\}.\]

In particular, if $H$ is empty then $Q(G)$ is either $\{()\}$ or empty; these two possible results are interpreted as the boolean values $\ttrue$ and $\ffalse$ respectively.

A query $Q_1$ is said to be \emph{contained} in a query $Q_2$, if for every graph $G$ we have $Q_1(G)\subseteq Q_2(G)$. This is denoted by $Q_1\subseteq Q_2$. 

If $B$ is the body of a conjunctive query with nonequalities, then $B^{\textrm{rel}}$ denotes the set of
relation atoms in $B$.
As is customary in the theory of conjunctive queries, we can view
the body of a conjunctive query without nonequalities as a graph
whose nodes are the variables.

Recall that a homomorphism is a matching from a body without nonequalities to another body without nonequalities, viewed as a graph.



\begin{lemma}\label{lem:cqcont}
Let $Q_1$: $H_1\leftarrow B_1$ and $Q_2$: $H_2 \leftarrow B_2$ be conjunctive queries with nonequalities. If $Q_1 \subseteq Q_2$ then there exists a homomorphism $h: B_2^{\textrm{rel}} \rightarrow B_1^{\textrm{rel}}$.\begin{proof}
  Notice that $H_1\in Q_1(B_1^{\textrm{rel}})$ since the identity map is clearly a matching. Hence $H_1 \in Q_2(B_1^{\textrm{rel}})$ because $Q_1\subseteq Q_2$ by hypothesis. Therefore there exists a matching $f:B_2 \rightarrow B_1^{\textrm{\textrm{rel}}}$, which is also a matching from $B_2^{\textrm{rel}}$ to $B_1^{\textrm{rel}}$, and is hence the desired homomorphism.
\end{proof}
\end{lemma}

We say that a directed graph $G$ is a \emph{chain} if it has no loops or cycles and its \emph{undirected} version is isomorphic to the undirected chain with nodes $1,\ldots,n$ where $n$ is the number of nodes of $G$. Such a chain has edges $\{i,i+1\}$ for $i = 1,\ldots,n-1$. Beware that in this terminology, a chain may have forward as well as backward edges, as illustrated in Figure~\ref{fig:chain}.
\begin{figure}
  \begin{center}
    \begin{tikzpicture}[scale=1.2,shorten >=0pt,->]
        \foreach \name in {1,...,7}
          \node[vertex] (\name) at (\name,1) {};

        \foreach \o/\t in {1/2,2/3,4/5}
            \draw[arrows={-stealth'}] (\o)  to [bend left] (\t);
        \foreach \o/\t in {4/3,6/5,7/6}
            \draw[arrows={-stealth'}] (\o)  to [bend right] (\t);
    \end{tikzpicture}
  \end{center}
  \caption{Example of a chain.\label{fig:chain}}
\end{figure}
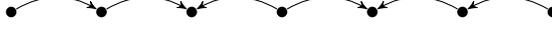

The following lemma can easily be proven by structural induction. 
\begin{lemma}\label{lem:invdiconjunctive}
    If $e$  is a union-free expression in $\elang{\conv{},\di}$, then there exists an equivalent conjunctive query $Q$: $H(x,y) \leftarrow B$ with nonequalities such that $B^{\textrm{rel}}$ has the form of a disjoint union of chains. 
\end{lemma}

Let $Q_{ZZZ}$ be the conjunctive query $() \leftarrow B_{ZZZ}$ that checks for the existence of the pattern displayed in Figure\nobreakspace \ref {fig:triplezigzag}. The name ZZZ is derived from the characteristic triple zigzag form of the pattern. For later use, we show the following. (Recall that an \emph{endomorphism} of a structure $A$ is a homomorphism from $A$ to itself.)
\begin{lemma}\label{lem:bzzzendo}
  The $B_{ZZZ}$ pattern has no endomorphism except for the identity.\end{lemma}\begin{proof} Let $f$ be an endomorphism of the $B_{ZZZ}$ pattern in Figure\nobreakspace \ref {fig:triplezigzag}. We first show that $f(a) = a$. Note that there has to start a directed path of length 6 in $f(a)$ for the homomorphism property to hold since there starts a directed path of length 6 in $a$. Therefore $f(a) = a$ or $f(a) = j$. If $f(a) = j$ then $f(g) = k$, and hence $f(j) = l$. This, however, is not possible since there starts a directed path of length 6 in $j$ but not in $l$. Therefore $f(a) = a$. 

    Now, the only thing left to verify is that no chain starting in $a$ can be mapped homomorphically on another chain starting in $a$. First note that every chain starting in $a$ has a very special structure, i.e., a path of forward edges, followed by an inverted edge, which is again followed by the same number of forward edges as before the inverted edge.  Therefore, it is clear that a chain $C_1$ starting in $a$ can only be mapped on another chain $C_2 \neq C_1$ starting in $a$, if and only if, the number of forward edges in $C_1$ minus one is at most the number of forward edges in $C_2$ preceding the inverted edge. In our graph, however, the number of forward edges in every chain starting in $a$ minus one is at least seven, and the number of forward edges in every chain starting in $a$ preceding the inverted edge is at most six. Therefore we can conclude that $f$ maps every node onto itself as desired.
  \end{proof}

\begin{figure}[tbp!]
\begin{center}
\begin{tikzpicture}[scale=1.2,shorten >=0pt,->]

  \node[vertex] (a) at (1,0) [label=-180:$a$]{};
  \foreach \name/\x in {cc/3,ff/6}
    \node[vertex] (\name) at (\x,1) {};
  
  \node[vertex] (bb) at (2,1) [label=90:$b$]{};
  \node[vertex] (b) at (2,0) [label=90:$c$]{};

    \foreach \name/\x in {gg/3,hh/4,ii/5,jj/6,kk/7,ll/8}
      \node[vertex] (\name) at (\x,1.7) {}; 
      \node[vertex] (bbb) at (2,-1) [label=90:$d$]{};
      \foreach \name/\x in {ccc/3,eee/5,fff/6,ggg/7}
        \node[vertex] (\name) at (\x,-1) {};
        
        \foreach \name/\x in {iii/5,jjj/6,kkk/7,lll/8,mmm/9,nnn/10}
          \node[vertex] (\name) at (\x,-0.5) {};

   \foreach \name/\x in {b/2,c/3}
      \node[vertex] (\name) at (\x,0){};

\node[vertex] (d) at (4,0) {};
  \node[vertex] (ee) at (5,1){};
    \node[vertex] (dd) at (4,1){};
  \node[vertex] (e) at (5,0){};
    \node[vertex] (ll) at (8,1.7) {};
    \node[vertex] (j) at (7,0.6) {};
    \node[vertex] (mmm) at (9,-0.5) [label=90:$l$]{};
    \node[vertex] (nnn) at (10,-0.5) [label=0:$k$]{};
    \node[vertex] (f) at (3,0.6) {};
  \node[vertex] (hhh) at (4,-0.5) [label=180:$j$]{};
  \node[vertex] (gg) at (3,1.7) {};
    \node[vertex] (ff) at (6,1){};
        \node[vertex] (fff) at (6,-1){};
        
        \node[vertex] (ddd) at (4,-1) {};
            \node[vertex] (eee) at (5,-1){};
      \node[vertex] (ggg) at (7,-1) [label=0:$g$]{};
    \foreach \name/\x in {f/3,g/4,h/5,i/6,j/7}
      \node[vertex] (\name) at (\x,0.6) {};

    \foreach \o/\t in {a/b,b/c,c/d,d/e,f/g,g/h,h/i,i/j,bb/cc,cc/dd,dd/ee,ee/ff,gg/hh,hh/ii,ii/jj,jj/kk,kk/ll,bbb/ccc,ccc/ddd,ddd/eee,eee/fff,fff/ggg,hhh/iii,iii/jjj,jjj/kkk,kkk/lll,lll/mmm,mmm/nnn}
    \draw[arrows={-stealth'}] (\o)  to [out=20,in=160] (\t);

      \draw[arrows={-stealth'}] (f)  to [out=-20,in=128] (e);
      \draw[arrows={-stealth'}] (hhh)  to [out=-35,in=128] (ggg);
      \draw[arrows={-stealth'}] (a)  to [out=-80,in=170] (bbb);
      
      \draw[arrows={-stealth'}] (gg)  to [out=-35,in=128] (ff);
      \draw[arrows={-stealth'}] (a)  to [out=80,in=190] (bb);
  \end{tikzpicture}
\caption{Query pattern $B_{ZZZ}$ used to prove
  Proposition\nobreakspace \ref {bottom-pi-tech}. All edges are
  assumed to have the same label $R$.}
\label{fig:triplezigzag}
\end{center}
\end{figure}
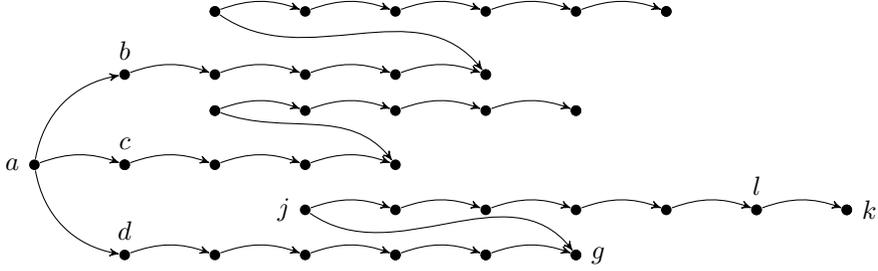

We are now ready to prove Proposition~\ref{bottom-pi-tech}.
\begin{proof}[Proof of Proposition\nobreakspace \ref {bottom-pi-tech}]
 The boolean query $Q_{ZZZ}$ is expressible in $\elang{\conv{},\pi}$ by 
  \[\pi_1(\rel^4 \circ \ler \circ \rel^4)\circ \pi_1(\rel^5 \circ \ler \circ \rel^5) \circ \pi_1(\rel^6 \circ \ler \circ \rel^6).\]
 This can be seen to be equivalent to 
 \[\pi_1(R^4 \circ \pi_2(\pi_1(R^4) \circ R)) \circ \pi_1(R^5 \circ \pi_2(\pi_1(R^5) \circ R)) \circ \pi_1(R^6 \circ \pi_2(\pi_1(R^6) \circ R))
 \] in $\elang{\pi}$ (a general argument for a result of this type will be given in the proof of Proposition\nobreakspace \ref {prop:converse-elimination}). 
   Let us now, for the sake of contradiction, assume that $Q_{ZZZ}$ is also expressible in $\elang{\conv{},\di}$ by an expression $Q$. 
   Hence, for every graph $G$: (1)~if $Q_{ZZZ}(G) = \ttrue$ then $Q(G) \neq \emptyset$, and (2)~if $Q(G) \neq \emptyset$ then $Q_{ZZZ}(G) = \ttrue$. Since unions in $\elang{\conv{},\di}$ can always be brought outside, we can assume that $Q = \bigcup_{i=0}^n e_i$ 
   for some $n \in \mathbb{N}$ where each $e_i$ is a union-free expression in $\elang{\conv{},\di}$. Now, since $Q_{ZZZ}(B_{ZZZ}) = \ttrue$, we also have $Q(B_{ZZZ}) = \cup_{i=0}^n e_i(B_{ZZZ}) \neq \emptyset$. Hence there exists $e \in \{e_0,\ldots,e_n\}$ such that $e(B_{ZZZ}) \neq \emptyset$. By Lemma\nobreakspace \ref {lem:invdiconjunctive}, $e$ is equivalent to a conjunctive query with nonequalities $H_{e} \leftarrow B_{e}$ such that $B_e^{\textrm{rel}}$ is a disjoint union of chains. Furthermore, since $e(B_{ZZZ})\neq \emptyset$ there exists a matching $f: B_{e}^{\textrm{rel}} \rightarrow B_{ZZZ}$ which is a homomorphism by definition.

   Now let $Q_e$ be the conjunctive query with nonequalities $() \leftarrow B_{e}$ so that $Q_e(G) = \ttrue$ if and only if $e(G) \neq \emptyset$ for every graph $G$. Since $e(G) \subseteq Q(G)$ for any graph $G$,  $Q_e(G) = \ttrue$ implies $Q(G) \neq \emptyset$, whence by (2) $Q_{ZZZ}(G) =\ttrue$. Therefore $Q_e \subseteq Q_{ZZZ}$. By Lemma\nobreakspace \ref {lem:cqcont} there is a homomorphism $g$ from $B_{ZZZ}$ into $B_e^{\textrm{rel}}$. Notice that in the $B_{ZZZ}$ pattern displayed in Figure\nobreakspace \ref {fig:triplezigzag}, the left most node, labeled $a$, has three outgoing edges. Furthermore, since $B_e^{\textrm{rel}}$ is a disjoint union of chains, no node in $B_e^{\textrm{rel}}$ has 3 outgoing edges, and hence two out of $g(b)$, $g(c)$ and $g(d)$ are equal. Thus $g$ is not injective.

   Now consider $g$ followed by $f$. This function is an endomorphism of $B_{ZZZ}$. Because $g$ is not injective, this endomorphism is not injective, and hence certainly not the identity, which contradicts Lemma\nobreakspace \ref {lem:bzzzendo}. Therefore $Q$ does not exist. 
\end{proof}
\section{Path queries}
\label{sec:path-queries}
In this section, we characterize the order $\ltpath$ of relative
expressiveness for path queries by
Theorem\nobreakspace \ref {th:path-inclusion} below.

Towards the statement of this characterization, first notice the
following interdependencies between features:
\begin{align*}
 \pi_1(e) & =  (e\circ\conv e)\cap\id =
(e\circ (\id \cup \div))\cap\id = \cpi_1(\cpi_1(e)); \\
 \pi_2(e) & =  (\conv e\circ e)\cap\id =
((\id \cup \div) \circ e)\cap\id = \cpi_2(\cpi_2(e)); \\
\cpi_1(e) & = 
\id\setminus\pi_1(e); \\
\cpi_2(e) & = 
\id\setminus\pi_2(e); \\
e_1 \cap e_2 & = e_1 \setminus (e_1 \setminus e_2).
\end{align*}

Notice that these rewriting rules with $e$ as their input
variable provide a means to translate an expression into an
equivalent expression in another language.

Inspired by the above interdependencies, for any set of nonbasic features $F$, we define $\fbar$ to be the smallest superset of $F$ satisfying the following rules:
\begin{itemize}
  \item If $\cpi \in \alang{F}$, then $\pi \in \alang{F}$;
  \item If $\cap \in \alang{F}$ and $\di \in \alang{F}$, then $\pi \in \alang{F}$;
  \item If $\cap \in \alang{F}$ and ${}^{-1}\in \alang{F}$, then $\pi \in \alang{F}$;
  \item If $\setminus \in \alang{F}$ and $\pi\in \alang{F}$, then $\cpi\in \alang{F}$.
  \item If $\setminus \in \alang{F}$, then $\cap \in \alang{F}$;
\end{itemize}
We can compute $\fbar$ from $F$ by repeated application of the above rules, a process which terminates quickly after at most three iterations.
For example, $\overline{ \{ \setminus, \conv{} \} } = \{ \setminus,
\conv{}, \cap, \pi, \cpi\}$.

Notice that, if $F_1 \subseteq \fbar_2$, we can always rewrite an
expression $e \in \BL(F_1)$ into an equivalent expression in
$\BL(F_2)$ using the rewriting rules displayed above. Notice that  Therefore, we obtain

\begin{proposition}
\label{th:path-if} 
If $F_1 \subseteq \fbar_2$, then $\BL(F_1) \ltpath \BL(F_2)$.
\end{proposition}

We will actually show that the converse also holds, whence
\begin{theorem} \label{th:path-inclusion} 
$\BL(F_1) \ltpath \BL(F_2)$ if and only if $F_1 \subseteq \fbar_2$.
\end{theorem}
The ``only if'' direction of Theorem\nobreakspace \ref {th:path-inclusion}
requires a detailed analysis. For clarity of presentation, we
divide the languages under consideration into two classes, i.e., the
class $\clbot$ of languages without intersection, and the class
$\clcap$ of languages with intersection. Formally:
\begin{align*}
  \clbot &= \{ \BL(F) \mid \cap \not\in \fbar \},\\
  \clcap &= \{ \BL(F) \mid \cap\,  \in \fbar\}.
 \end{align*}
 We first establish the ``only if'' direction for the cases where
 $\BL(F_1)$ and $\BL(F_2)$ belong to the same class.  We do so for
 each class separately in Sections\nobreakspace \ref {sec:languages-bottom} and\nobreakspace  \ref {sec:languages-cap}.  Finally, in
 Section\nobreakspace \ref {sec:path-separ-betw-subl}, we consider the case where
 $\BL(F_1)$ and $\BL(F_2)$ belong to distinct classes.

\subsection{Languages without $\cap$}
\label{sec:languages-bottom}
In this subsection, we show the ``only if'' direction of
Theorem\nobreakspace \ref {th:path-inclusion}, restricted to $\clbot$, the class of
languages without $\cap$. Stated positively, the proposition states that for fragments $F_1$ and $F_2$ using only the operators $\di$, $\pi$ and $\cpi$, $\elang{F_1}\ltpath \elang{F_2}$ can only hold if $F_1 \subseteq \alang{F_2}$.

\begin{proposition} \label{prop:path-bottom-Hasse} Let $\BL(F_1)$ and
  $\BL(F_2)$ be in $\clbot$. If $F_1 \not \subseteq \fbar_2$, then
  $\BL(F_1) \not \ltpath \BL(F_2)$.
\end{proposition}

Propositions\nobreakspace \ref {th:path-if} and\nobreakspace  \ref {prop:path-bottom-Hasse}
combined yield the Hasse diagram of $\ltpath$ for $\clbot$, shown in
Figure\nobreakspace \ref {fig:hasse-noint}. It is indeed readily verified that for
any two languages $\BL(F_1)$ and $\BL(F_2)$ in $\clbot$, there is a
path from $\BL(F_1)$ to $\BL(F_2)$ in Figure\nobreakspace \ref {fig:hasse-noint} if
and only if $F_1 \subseteq \fbar_2$.

\begin{figure*}[tbp]
\begin{center}
    \resizebox{0.67\textwidth}{!}{
\begin{picture}(0,0)%
\includegraphics{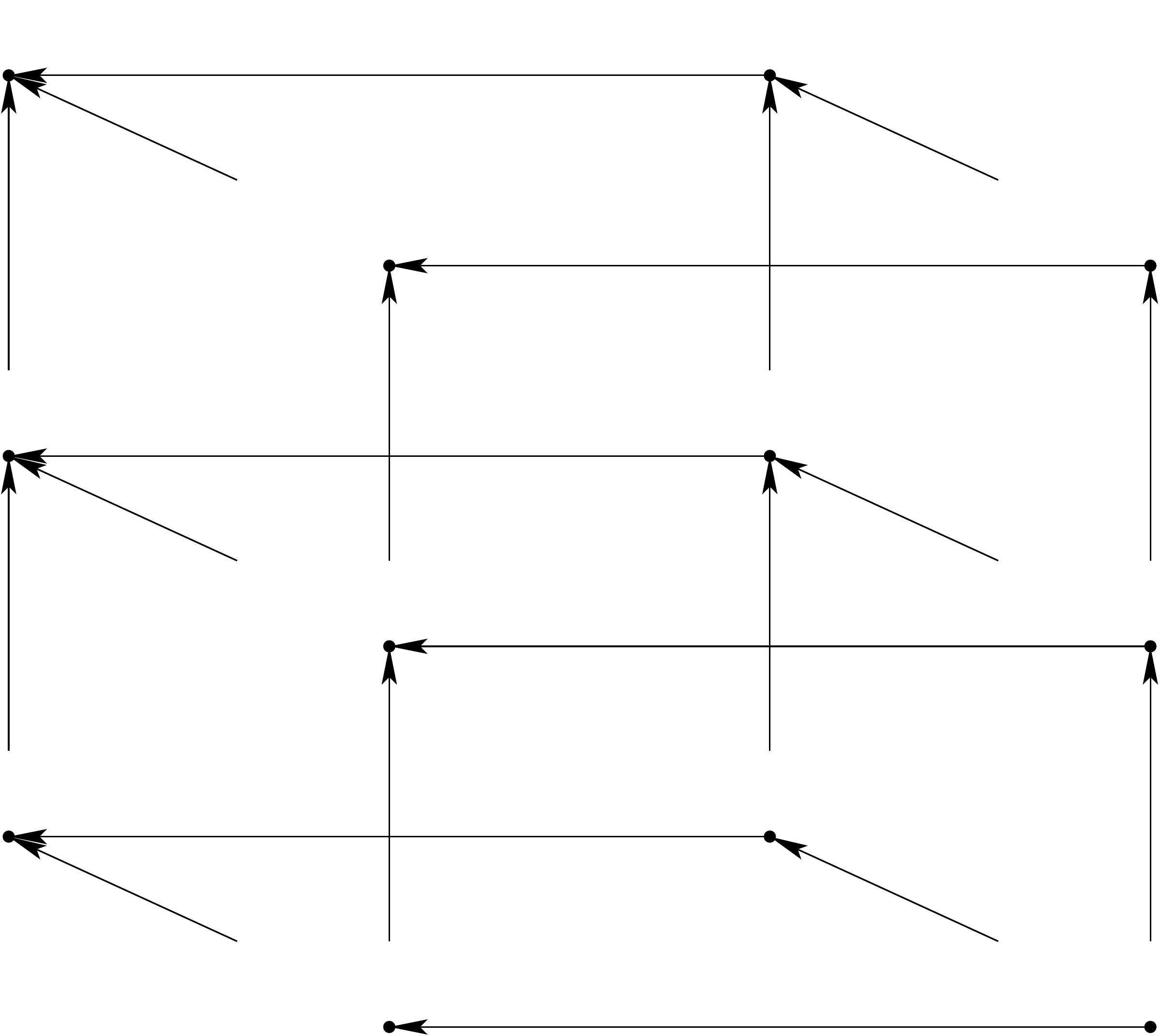}%
\end{picture}%
\setlength{\unitlength}{4144sp}%
\begingroup\makeatletter\ifx\SetFigFont\undefined%
\gdef\SetFigFont#1#2#3#4#5{%
  \reset@font\fontsize{#1}{#2pt}%
  \fontfamily{#3}\fontseries{#4}\fontshape{#5}%
  \selectfont}%
\fi\endgroup%
\begin{picture}(10964,9793)(-81,-10043)
\put(  1,-7801){\makebox(0,0)[b]{\smash{{\SetFigFont{25}{30.0}{\familydefault}{\mddefault}{\updefault}{\color[rgb]{0,0,0}$\BL(\,\framebox[1.3\width]{\hhh${}^{-1},\textit{di}$}\,)$}%
}}}}
\put(  1,-4201){\makebox(0,0)[b]{\smash{{\SetFigFont{25}{30.0}{\familydefault}{\mddefault}{\updefault}{\color[rgb]{0,0,0}$\BL(\,\framebox[1.2\width]{\hhh${}^{-1},\textit{di},\pi$}\,)$}%
}}}}
\put(3601,-9601){\makebox(0,0)[b]{\smash{{\SetFigFont{25}{30.0}{\familydefault}{\mddefault}{\updefault}{\color[rgb]{0,0,0}$\BL(\,\framebox[1.7\width]{\hhh${}^{-1}$}\,)$}%
}}}}
\put(3601,-2401){\makebox(0,0)[b]{\smash{{\SetFigFont{25}{30.0}{\familydefault}{\mddefault}{\updefault}{\color[rgb]{0,0,0}$\BL(\,\framebox[1.3\width]{\hhh${}^{-1},\overline\pi$}\,,\pi)$}%
}}}}
\put(3601,-6001){\makebox(0,0)[b]{\smash{{\SetFigFont{25}{30.0}{\familydefault}{\mddefault}{\updefault}{\color[rgb]{0,0,0}$\BL(\,\framebox[1.3\width]{\hhh${}^{-1},\pi$}\,)$}%
}}}}
\put(  1,-601){\makebox(0,0)[b]{\smash{{\SetFigFont{25}{30.0}{\familydefault}{\mddefault}{\updefault}{\color[rgb]{0,0,0}$\BL(\,\framebox[1.2\width]{\hhh${}^{-1},\textit{di},\overline\pi$}\,,\pi)$}%
}}}}
\put(7201,-7801){\makebox(0,0)[b]{\smash{{\SetFigFont{25}{30.0}{\familydefault}{\mddefault}{\updefault}{\color[rgb]{0,0,0}$\BL(\,\framebox[1.7\width]{\hhh$\textit{di}$}\,)$}%
}}}}
\put(7201,-4201){\makebox(0,0)[b]{\smash{{\SetFigFont{25}{30.0}{\familydefault}{\mddefault}{\updefault}{\color[rgb]{0,0,0}$\BL(\,\framebox[1.3\width]{\hhh$\textit{di},\pi$}\,)$}%
}}}}
\put(10801,-9601){\makebox(0,0)[b]{\smash{{\SetFigFont{25}{30.0}{\familydefault}{\mddefault}{\updefault}{\color[rgb]{0,0,0}$\BL$}%
}}}}
\put(10801,-6001){\makebox(0,0)[b]{\smash{{\SetFigFont{25}{30.0}{\familydefault}{\mddefault}{\updefault}{\color[rgb]{0,0,0}$\BL(\,\framebox[1.7\width]{\hhh$\pi$}\,)$}%
}}}}
\put(10801,-2401){\makebox(0,0)[b]{\smash{{\SetFigFont{25}{30.0}{\familydefault}{\mddefault}{\updefault}{\color[rgb]{0,0,0}$\BL(\,\framebox[1.7\width]{\hhh$\overline\pi$}\,,\pi)$}%
}}}}
\put(7201,-601){\makebox(0,0)[b]{\smash{{\SetFigFont{25}{30.0}{\familydefault}{\mddefault}{\updefault}{\color[rgb]{0,0,0}$\BL(\framebox[1.3\width]{\hhh$\textit{di},\overline\pi$}\,,\pi)$}%
}}}}
\end{picture}%
      }
    \label{fig:hasse-noint}
\end{center}
\caption{The Hasse diagram of $\ltpath$ for $\clbot$. For each language, the boxed features are a minimal set of
  nonbasic features defining the language, while the other features
  can be derived from them in the sense of
  Theorem\nobreakspace \ref {th:path-inclusion} (using the appropriate
  interdependencies).}
\end{figure*}

Towards a proof of Proposition\nobreakspace \ref {prop:path-bottom-Hasse}, we first
establish an auxiliary proposition.  For later use, we sometimes prove results
that are stronger than strictly needed for this purpose.

\begin{proposition} \label{prop:path-bottom}
  Let $F_1$ and $F_2$ be sets of nonbasic features.
  \begin{enumerate}
  \item \label{bottom-di} If $\div \in \fbar_1$ and $\div \not\in \fbar_2$, then $\BL(F_1)
    \nltstrong \BL(F_2)$.
  \item \label{bottom-cpi} If $\cpi \in \fbar_1$, $\cpi \not\in
    \fbar_2$, and $\setminus \not\in \fbar_2$, then $\BL(F_1) \nltstrong
    \BL(F_2)$.
  \item \label{bottom-conv} If $\conv{} \in \fbar_1$ and $\conv{} \not\in
    \fbar_2$, then $\BL(F_1) \nltstrongpath \BL(F_2)$.
  \item \label{bottom-pi} If $\pi \in \fbar_1$ and $F_2 \subseteq \{\conv{}, \div \}$, then $\BL(F_1) \not\ltbool \BL(F_2)$.
  \end{enumerate}
\end{proposition}

\begin{proof}
  For (\ref{bottom-di}), consider a graph $G_1$ consisting of two
  self-loops, and a graph $G_2$ consisting of a single self-loop, all with
  the same label. For any nontrivial expression $e$ not using $\di$, $\setminus$ or $\cpi$, it is evident that $e(G_1)$ and $e(G_2)$ both contain all possible self-loops in $G_1$ and $G_2$ respectively. Therefore, applying $\setminus$ or $\cpi$ to any such expressions leads to expressions that show similar behavior on $G_1$ and $G_2$. More specifically, they select either all self-loops in both $G_1$ and $G_2$, or select nothing in both graphs simultaneously. The same reasoning can now be applied to general expressions in $\elang{F_2}$. This reasoning shows that $G_1$ and $G_2$ cannot be distinguished in $\elang{F_2}$. They are, however, distinguishable by in $\elang{F_1}$ by $\di \neq \emptyset$.

  For (\ref{bottom-cpi}), notice that $\alang{F_2}\subseteq \{\di,\pi,\cap,{}^{-1},{}^+\}$, whence $\elang{F_2}$ only contains monotone expressions. Therefore it is clear that a non-monotone query such as $\cpi_2(R)\neq \emptyset$ is not expressible in $\elang{F_2}$.

%

  For (\ref{bottom-conv}), we establish strong  separation at the
  level of path queries as explained in
  Section\nobreakspace \ref {sec:strong-path-separation}.
  Thereto, we consider the
  graph $G$ shown in Figure\nobreakspace \ref {fig:converse}. By the Brute-Force method described in Section~\ref{sec:brute-force-approach}. we can exhaustively enumerate all the possible result relations $e(G)$ for all expressions $e \in \elang{\di,\setminus,{}^+}$, i.e., not using converse. There are 128 relations in this list. It can then be verified that $G^{-1}$ is not present in the list\footnote{Note that if we would have used a simpler graph $G$, say $G$ consisting of a single edge, then $G^{-1}$ would be expressible without using converse, using the expression $\di \setminus R$}. 

  The proof of (\ref{bottom-pi}) follows directly from Proposition~\ref{bottom-pi-tech} since $\elang{\pi} \pathq\elang{F_1}$ and $\elang{F_2} \pathq \elang{\conv{},\di}$.
\end{proof} 

\begin{figure}[tbp!]
\begin{center}
 \resizebox{0.425\textwidth}{!}{
\includegraphics{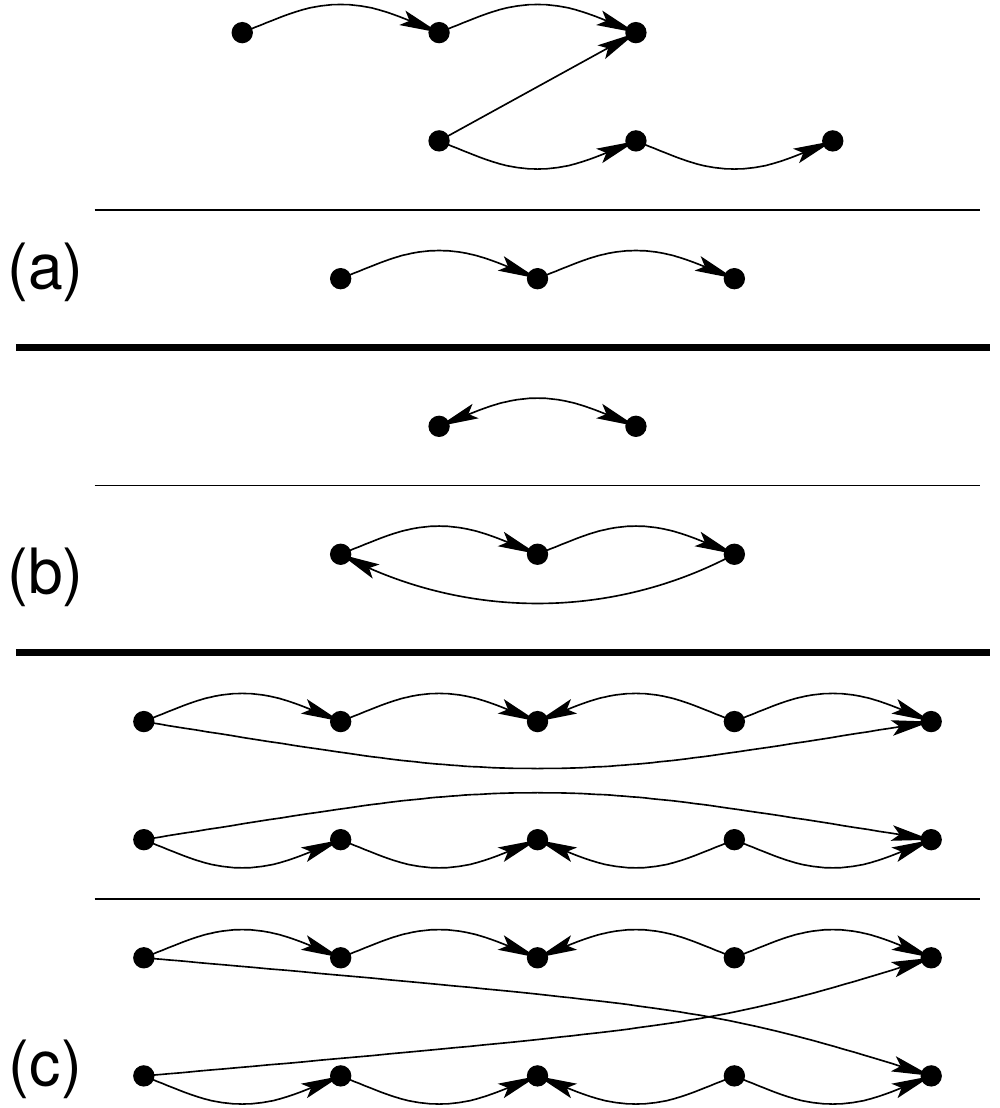}}
\caption{Graph pairs used to prove $\nltstrong$ results in
  Section\nobreakspace \ref {sec:path-queries} and \ref{sec:boolean-queries}. All
  edges are assumed to have the same label $R$.}
\label{fig:separations}
\end{center}
\end{figure}

\begin{figure}[tbp!]
\begin{center}
 \resizebox{0.4\textwidth}{!}{
\includegraphics{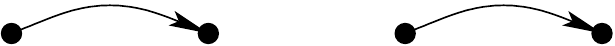}}
\caption{Graph used to prove Proposition\nobreakspace \ref {prop:path-bottom}
  (\ref{bottom-conv}). Both edges are assumed to have the same label
  $R$.}
\label{fig:converse}
\end{center}
\end{figure}

Proposition\nobreakspace \ref {prop:path-bottom} is now used to show that for every
pair $F_1$ and $F_2$ of sets of nonbasic features for which $F_1
\not\subseteq \fbar_2$ (i.e., for which there is no path in
Figure\nobreakspace \ref {fig:hasse-noint}), that $\BL(F_1) \not\ltpath \BL(F_2)$.
The remainder of the proof of Proposition\nobreakspace \ref {prop:path-bottom-Hasse} is a combinatorial analysis to verify that Proposition\nobreakspace \ref {prop:path-bottom} covers all the cases.  

\begin{proof}[Proof of Proposition\nobreakspace \ref {prop:path-bottom-Hasse}]
    First, suppose that $\cpi \in F_2$. Then, $F_1 \not\subseteq \alang{F_2}$ if and only if $F_1 \cap \{\di,{}^{-1}\} \not\subseteq F_2 \cap \{\di,{}^{-1}\}$. Hence we have the following possible scenarios: $\di \in F_1$ and $\di \not\in \alang{F_2}$; or ${}^{-1} \in F_1$ and ${}^{-1} \not\in \alang{F_2}$. If $\di \in F_1$ and $\di \not\in \alang{F_2}$, then $\elang{F_1} \not\pathq \elang{F_2}$ due to Proposition\nobreakspace \ref {prop:path-bottom}(\ref{bottom-di}). Otherwise, we achieved the result due to Proposition\nobreakspace \ref {prop:path-bottom}(\ref{bottom-conv}).
  
    On the other hand, suppose that $\cpi \not \in F_2$. Then, $F_2 = \alang{F_2}$. Thus, $F_1 \not \subseteq \alang{F_2}$ if and only if $F_1 \not \subseteq F_2$. Hence there has to exists some $x \in F_1$ such that $x \not \in F_2$. Furthermore, since $F_1 \subseteq \{\di,\pi,\cpi,^{-1}\}$ and $F_2 \subseteq \{\di,\pi,{}^{-1}\}$ Proposition\nobreakspace \ref {prop:path-bottom} can be applied. Notice that we cannot apply this proposition directly since it makes use of $\alang{F_1}$ instead of $F_1$. This, however, is no issue since $F_1 \subseteq \alang{F_1}$.
\end{proof}

\subsection{Languages with $\cap$}
\label{sec:languages-cap}
In this subsection, we show the ``only if'' direction of
Theorem\nobreakspace \ref {th:path-inclusion}, restricted to $\clcap$, the class of
languages with $\cap$.

\begin{proposition} \label{prop:path-int-Hasse} Let both $\BL(F_1)$ and
  $\BL(F_2)$ be in $\clcap$. If $F_1 \not \allowbreak\subseteq \fbar_2$, then
  $\BL(F_1) \not \ltpath \BL(F_2)$.
\end{proposition}
\begin{figure*}[tbp]
\begin{center}
    \resizebox{0.72\textwidth}{!}{
\begin{picture}(0,0)%
\includegraphics{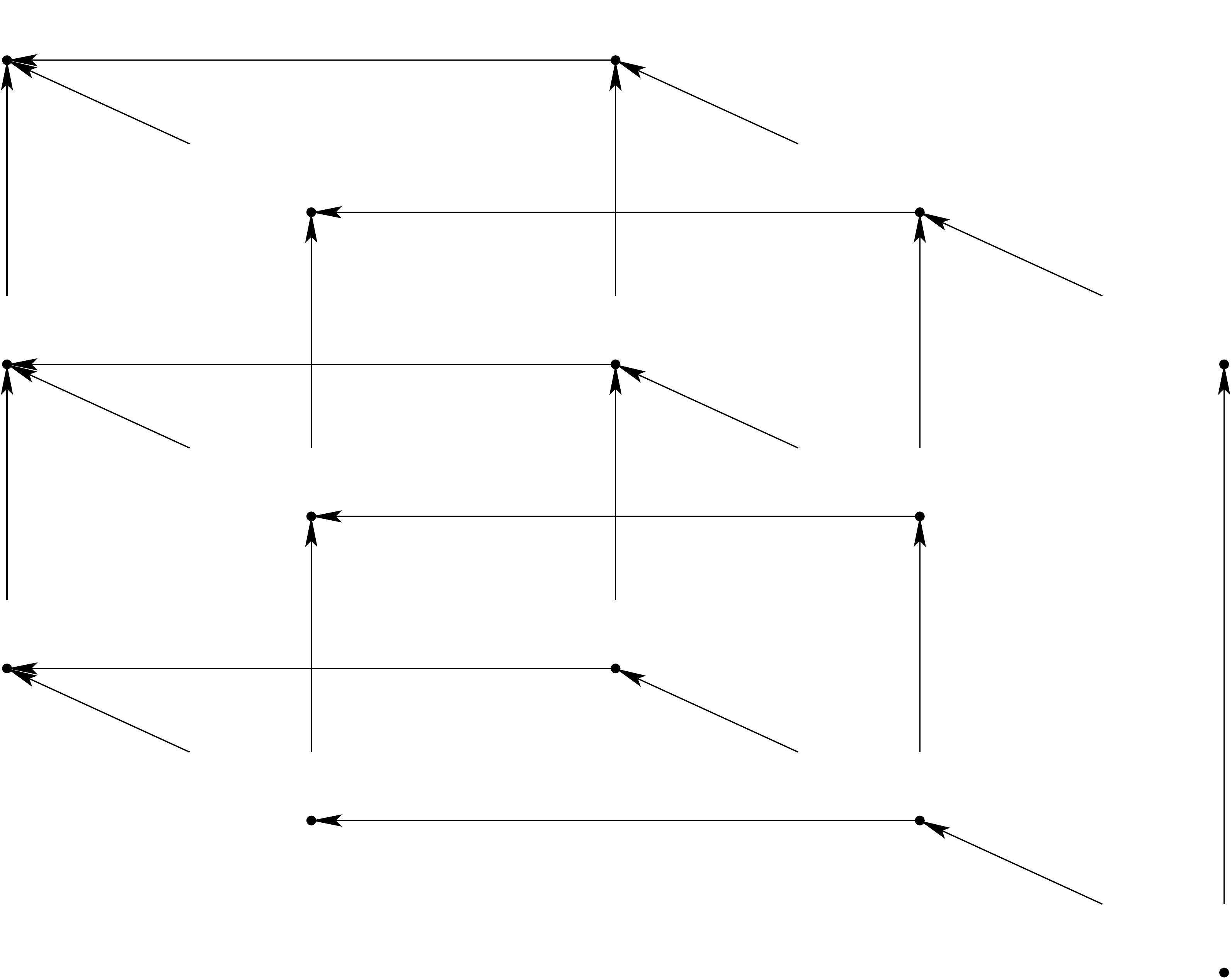}%
\end{picture}%
\setlength{\unitlength}{4144sp}%
\begingroup\makeatletter\ifx\SetFigFont\undefined%
\gdef\SetFigFont#1#2#3#4#5{%
  \reset@font\fontsize{#1}{#2pt}%
  \fontfamily{#3}\fontseries{#4}\fontshape{#5}%
  \selectfont}%
\fi\endgroup%
\begin{picture}(14564,11572)(-81,-11822)
\put(7201,-4201){\makebox(0,0)[b]{\smash{{\SetFigFont{25}{30.0}{\familydefault}{\mddefault}{\updefault}{\color[rgb]{0,0,0}$\BL(\,\framebox[1.2\width]{\hhh$\cap,\textit{di},\overline\pi$}\,,\pi)$}%
}}}}
\put(  1,-4201){\makebox(0,0)[b]{\smash{{\SetFigFont{25}{30.0}{\familydefault}{\mddefault}{\updefault}{\color[rgb]{0,0,0}$\BL(\,\framebox[1.15\width]{\hhh$\cap,{}^{-1},\textit{di},\overline\pi$}\,,\pi)$}%
}}}}
\put(  1,-7801){\makebox(0,0)[b]{\smash{{\SetFigFont{25}{30.0}{\familydefault}{\mddefault}{\updefault}{\color[rgb]{0,0,0}$\BL(\,\framebox[1.2\width]{\hhh$\cap,{}^{-1},\textit{di}$}\,,\pi)$}%
}}}}
\put(10801,-2401){\makebox(0,0)[b]{\smash{{\SetFigFont{25}{30.0}{\familydefault}{\mddefault}{\updefault}{\color[rgb]{0,0,0}$\BL(\,\framebox[1.3\width]{\hhh$\setminus,\pi$}\,,\cap,\overline\pi)\makebox[0pt][l]{$\mathstrut=\BL(\,\framebox[1.3\width]{\hhh$\setminus,\overline\pi$}\,,\cap,\pi)$}$}%
}}}}
\put(3601,-9601){\makebox(0,0)[b]{\smash{{\SetFigFont{25}{30.0}{\familydefault}{\mddefault}{\updefault}{\color[rgb]{0,0,0}$\BL(\,\framebox[1.3\width]{\hhh$\cap,{}^{-1}$}\,,\pi)$}%
}}}}
\put(3601,-6001){\makebox(0,0)[b]{\smash{{\SetFigFont{25}{30.0}{\familydefault}{\mddefault}{\updefault}{\color[rgb]{0,0,0}$\BL(\,\framebox[1.2\width]{\hhh$\cap,{}^{-1},\overline\pi$}\,,\pi)$}%
}}}}
\put(3601,-2401){\makebox(0,0)[b]{\smash{{\SetFigFont{25}{30.0}{\familydefault}{\mddefault}{\updefault}{\color[rgb]{0,0,0}$\BL(\,\framebox[1.3\width]{\hhh$\setminus,{}^{-1}$}\,,\cap,\pi,\overline\pi)$}%
}}}}
\put(10801,-9601){\makebox(0,0)[b]{\smash{{\SetFigFont{25}{30.0}{\familydefault}{\mddefault}{\updefault}{\color[rgb]{0,0,0}$\BL(\,\framebox[1.3\width]{\hhh$\cap,\pi$}\,)$}%
}}}}
\put(14401,-11401){\makebox(0,0)[b]{\smash{{\SetFigFont{25}{30.0}{\familydefault}{\mddefault}{\updefault}{\color[rgb]{0,0,0}$\BL(\,\framebox[1.7\width]{\hhh$\cap$}\,)$}%
}}}}
\put(14401,-4201){\makebox(0,0)[b]{\smash{{\SetFigFont{25}{30.0}{\familydefault}{\mddefault}{\updefault}{\color[rgb]{0,0,0}$\BL(\,\framebox[1.7\width]{\hhh$\setminus$}\,,\cap)$}%
}}}}
\put(7201,-7801){\makebox(0,0)[b]{\smash{{\SetFigFont{25}{30.0}{\familydefault}{\mddefault}{\updefault}{\color[rgb]{0,0,0}$\BL(\,\framebox[1.3\width]{\hhh$\cap,\textit{di}$}\,,\pi)$}%
}}}}
\put(  1,-601){\makebox(0,0)[b]{\smash{{\SetFigFont{25}{30.0}{\familydefault}{\mddefault}{\updefault}{\color[rgb]{0,0,0}$\BL(\,\framebox[1.2\width]{\hhh$\setminus,{}^{-1},\textit{di}$}\,,\cap,\pi,\overline\pi)$}%
}}}}
\put(10801,-6001){\makebox(0,0)[b]{\smash{{\SetFigFont{25}{30.0}{\familydefault}{\mddefault}{\updefault}{\color[rgb]{0,0,0}$\BL(\,\framebox[1.3\width]{\hhh$\cap,\overline\pi$}\,,\pi)$}%
}}}}
\put(7201,-601){\makebox(0,0)[b]{\smash{{\SetFigFont{25}{30.0}{\familydefault}{\mddefault}{\updefault}{\color[rgb]{0,0,0}$\BL(\framebox[1.3\width]{\hhh$\setminus,\textit{di}$}\,,\cap,\pi,\overline\pi)$}%
}}}}
\end{picture}%
      }

\end{center}
\caption{The Hasse diagram of $\ltpath$ and $\ltbool$ for $\clcap$.\label{fig:hasse-int}}
\end{figure*}

Propositions\nobreakspace \ref {th:path-if} and\nobreakspace  \ref {prop:path-int-Hasse} combined
yield the Hasse diagram of $\ltpath$ for $\clcap$, shown in
Figure\nobreakspace \ref {fig:hasse-int}.

Towards a proof of Proposition\nobreakspace \ref {prop:path-int-Hasse}, we first
establish the following.

\begin{proposition} \label{prop:path-int}
  Let $F_1$ and $F_2$ be sets of nonbasic features.
  \begin{enumerate}
  \item \label{int-diff} If $\setminus \in \fbar_1$ and $\setminus \not\in
    \fbar_2$, then $\BL(F_1) \nltstrong \BL(F_2)$.
  \item \label{int-pi} If $\pi \in \fbar_1$, and $F_2 \subseteq
    \{\setminus, \cap\}$, then $\BL(F_1) \nltstrong \BL(F_2)$.
    \end{enumerate}
\end{proposition}

\begin{proof}
  For (\ref{int-diff}), consider a 3-clique $G_1$, and a bow-tie $G_2$ consisting
  of two 3-cliques (both graphs contain a self-loop on every node). It can be proven by straightforward induction and case analysis that for any nontrivial expression $e \in \elang{\di,{}^{-1},\cap,{}^+}$ at least $\id(G_i) \subseteq e(G_i)$ or $R \setminus \id(G_i) \subseteq e(G_i)$. In either case it is clear that a projection of any nontrivial expression in $\elang{\di,{}^{-1},\cap,{}^+}$ evaluated on both graphs leads to all self-loops. Using this fact, it can be seen that a coprojection of any expression in $\elang{\di,{}^{-1},\cap,\cpi,\pi,{}^+}$ leads to either all self-loops, or a completely empty query result on both graphs simultaneously. Therefore, no expression in $\elang{\di,{}^{-1},\cap,\cpi,\pi,{}^+}$ can distinguish $G_1$ and $G_2$. 
  The graphs, however, are
  distinguishable by the boolean query expressed by $R^2 \setminus
  R$. 

   For (\ref{int-pi}), consider the graphs displayed in Figure\nobreakspace \ref {fig:separations}~(a). Notice that expressions in $\BL$ select paths of the same length in both graphs simultaneously, e.g., if an expression selects all paths of length two in one graph, it also selects all the paths of length two in the other and vice versa. Therefore, expressions using set difference evaluate to empty or nonempty on both graphs simultaneously. Thus, expressions in $\elang{\setminus}$ cannot distinguish the considered graphs, whence they are indistinguishable in $\elang{F_2}$ as well since $\elang{F_2}\ltbool \elang{\setminus}$. The graphs, however, are
  distinguishable in $\elang{F_1}$ by the boolean query expressed by $\pi_1(R^2) \circ
  R \circ \pi_2(R^2)$.
\end{proof}

Propositions\nobreakspace \ref {prop:path-bottom} and\nobreakspace  \ref {prop:path-int} are now
used to show that for every pair $F_1$ and $F_2$ of sets of nonbasic
features for which $F_1 \not\subseteq \fbar_2$ (i.e., for which there
is no path in Figure\nobreakspace \ref {fig:hasse-int}), that $\BL(F_1) \not\ltpath
\BL(F_2)$.

The remainder of the proof of Proposition\nobreakspace \ref {prop:path-int-Hasse} is
a combinatorial analysis to verify that
Propositions\nobreakspace \ref {prop:path-bottom} and\nobreakspace  \ref {prop:path-int} cover all
relevant cases.
\begin{proof}[Proof of Proposition\nobreakspace \ref {prop:path-int-Hasse}]
By definition $\cap \in \alang{F_1}$ and $\cap \in \alang{F_2}$ since both $\elang{F_1}$ and $\elang{F_2}$ are in $\mathcal{C}[\cap]$. Hence, $F_1 \not\subseteq \alang{F_2}$ if and only if there exists $x \in \{ \pi,\cpi,\di,{}^{-1},\setminus\}$ such that $x \in F_1$ and $x \not\in \alang{F_2}$. We will consider every such $x$ and show that our result directly follows from Propositions~\ref{prop:path-bottom} or \ref{prop:path-int}.

If $x = \di$, $x = \conv{}$ or $x = \setminus$, then respectively  Proposition~\ref{prop:path-bottom}(\ref{bottom-di}), \ref{prop:path-bottom}(\ref{bottom-conv}) or \ref{prop:path-int}(\ref{int-diff}) gives us the desired result. 

If $x = \pi$, then clearly $\pi \not \in \alang{F_2}$ if and only if $F_2 \cap \{\di,{}^{-1},\cpi,\pi\} = \emptyset$. Hence $F_2 \subseteq \{\cap,\setminus\}$. Now, we can apply Proposition\nobreakspace \ref {prop:path-int}(\ref{int-pi}), which proves the result.

If $x = \cpi$, then using the interdependencies introduced in the beginning of Section\nobreakspace \ref {sec:path-queries} we get 
 \[\cpi \not \in \alang{F_2} \iff \setminus \not\in F_2\lor (\setminus \in F_2 \land \pi \not\in \alang{F_2}).\]
      So we have two scenarios. If $\setminus \not\in F_2$ then we can apply Proposition\nobreakspace \ref {prop:path-bottom}(\ref{bottom-cpi}) to prove our result. On the other hand, when $\setminus \in F_2$ we cannot apply Proposition\nobreakspace \ref {prop:path-bottom}(\ref{bottom-cpi}). As said above, now $\pi$ cannot be in $\alang{F_2}$. Furthermore, note that in this scenario
      \begin{align*}
        \setminus \in F_2\land \pi \not \in \alang{F_2} \iff F_2 \cap \{\conv{},\di\} = \emptyset
      \end{align*}
        which implies that $F_2 \subseteq \{\cap, \setminus\}$. Moreover, $\pi \in \alang{F_1}$ since $\cpi \in \alang{F_1}$.  Hence, we can apply proposition\nobreakspace \ref {prop:path-int}(\ref{int-pi}), which proves the result.
\end{proof}

\subsection{Cross-relationships between  subdiagrams}
\label{sec:path-separ-betw-subl}
To finish the proof of Theorem\nobreakspace \ref {th:path-inclusion}, we finally
show the ``only if'' direction for the case where $\BL(F_1)$ and
$\BL(F_2)$ belong to different classes.

\begin{proposition} \label{prop:path-cross-Hasse} Let $\BL(F_1)$ and
  $\BL(F_2)$ be languages such that one language belongs to $\clbot$,
  and the other language belongs to $\clcap$. If $F_1 \not \subseteq
  \fbar_2$, then $\BL(F_1) \not \ltpath \BL(F_2)$.
\end{proposition}

Towards a proof of Proposition\nobreakspace \ref {prop:path-cross-Hasse}, we first
establish the following.

\begin{proposition} \label{prop:path-cross-comp} Let $F_1$ and $F_2$
  be sets of nonbasic features.  If $\cap \in \fbar_1$ and $\cap
  \not\in \fbar_2$, then $\BL(F_1) \nltstrong \BL(F_2)$.
\end{proposition}

\begin{proof}
Since $\cap \not \in \alang{F_2}$ it must be that $F_2 \subseteq \{\di,{}^{-1},\pi,\cpi,{}^+\}$. So, it is sufficient to find a boolean query expressible in $\elang{F_1}$, which is not expressible in $\elang{\di,{}^{-1},\cpi,{}^+}$. Consider the graphs $G_1$ and $G_2$ in Figure\nobreakspace \ref {fig:separations}~(b). Notice that there starts and ends a path of every length in each node in both graphs. Utilizing this fact, it can be shown that for any nontrivial expression $e \in \elang{\di,{}^{-1},{}^+}$, it must be that $\pi_i(e)(G_j) = \id(G_j)$. Using this, it can be seen that the coprojection of any expression in $\elang{\di,{}^{-1},\cpi,{}^+}$ leads to either all self-loops, or a completely empty query result on both graphs simultaneously. Therefore, no expression in $\elang{\di,{}^{-1},\cpi,{}^+}$ can distinguish $G_1$ and $G_2$. The graphs, however, are distinguishable by the boolean
  query expressed by $R^2 \cap \id$. 
\end{proof}

As detailed below, Propositions\nobreakspace \ref {prop:path-bottom},  \ref {prop:path-int} and\nobreakspace  \ref {prop:path-cross-comp} are now subsequently
used to show that for every pair $F_1$ and $F_2$ of sets of nonbasic
features for which $F_1 \not\subseteq \fbar_2$, that $\BL(F_1)
\not\ltpath \BL(F_2)$, in the same way as in
Section\nobreakspace \ref {sec:languages-bottom} and \ref{sec:languages-cap}.

The remainder of the proof of Proposition\nobreakspace \ref {prop:path-cross-Hasse}
is again a combinatorial analysis to verify that the above-mentioned
propositions cover all relevant cases.  

\begin{proof}[Proof of Proposition\nobreakspace \ref {prop:path-cross-Hasse}]
  First, suppose that $\elang{F_1} \in \clcap$ and $\elang{F_2} \in \clbot$. Then, by definition $\cap \in \alang{F_1}$ and $\cap \not \in \alang{F_2}$. The result now follows directly from Proposition\nobreakspace \ref {prop:path-cross-comp}.

  On the other hand, suppose that $\elang{F_1}$ is in $\clbot$ and $\elang{F_2}$ is in $\clcap$. 
  Clearly, then $F_1 \nsubseteq \alang{F_2}$ if and only if 
  $F_1 \nsubseteq \alang{F_2} \setminus \{\cap,\setminus\}$. 
  Hence at least one feature $x$ of $\di,\pi,\cpi,{}^{-1}$ is present in $F_1$ but 
  missing in $\alang{F_2}$. We will consider every such $x$ and show that 
  our result directly follows  from Propositions~\ref{prop:path-bottom}, or \ref{prop:path-int}.

  If $x = \di$ or $x = \conv{}$, then respectively Proposition~\ref{prop:path-bottom}(\ref{bottom-di})
  or \ref{prop:path-bottom}(\ref{bottom-conv}) gives us the desired result.

  If $x = \pi$ then $\cpi \not \in F_2$ by the interdependencies introduced 
  in the beginning of Section\nobreakspace \ref {sec:path-queries}. Furthermore, 
  $\alang{F_2}\cap \{\conv{},\di\}= \emptyset$ since by 
  hypothesis $\cap \in \alang{F_2}$. Therefore $F_2\subseteq \{\setminus,\cap\}$, 
  and hence Proposition\nobreakspace \ref {prop:path-int}(\ref{int-pi}) can be applied, which proves the result. 

  If $x = \cpi$ then $\alang{F_2} \cap \{\setminus,\pi\} \neq \{\setminus,\pi\}$. 
  Suppose that $\setminus \not \in \alang{F_2}$, then our result follows from 
  Proposition\nobreakspace \ref {prop:path-bottom}(\ref{bottom-cpi}). On the other hand, if
  $\pi\not\in \alang{F_2}$, then the result follows from the previous 
  case since $\pi \in \alang{F_1}$.
\end{proof}


Propositions\nobreakspace \ref {th:path-if},  \ref {prop:path-bottom-Hasse},  \ref {prop:path-int-Hasse} and\nobreakspace  \ref {prop:path-cross-Hasse}, together
prove Theorem\nobreakspace \ref {th:path-inclusion}.

Hence, the Hasse diagram of $\ltpath$ can be obtained from the
subdiagrams for $\clbot$ and $\clcap$ by simply adding the 12 canonical inclusion arrows between the subdiagram for $\clbot$ and the
subdiagram for $\clcap$.  However, in the presence of $\cap$, $\di$ or $\conv{}$ gives $\pi$, so the arrows from $\elang{\di}$ to $\elang{\cap,\di,\pi}$, $\elang{\conv{}}$ to $\elang{\cap,\conv{},\pi}$, and $\elang{\conv{},\di}$ to $\elang{\cap,\conv{},\di,\pi}$ are transitive, and can therefore be omitted.

So, all paths between the subdiagrams are
induced by these canonical inclusion arrows and the 5 equations from
the beginning of Section\nobreakspace \ref {sec:path-queries}.

\section{Boolean queries}
\label{sec:boolean-queries}
In this section, we characterize the order $\ltbool$ of relative
expressiveness for boolean queries by
Theorem\nobreakspace \ref {th:bool-inclusion} below.

Towards the statement of this characterization, first observe that
$\BL(F_1) \ltpath \BL(F_2)$ implies $\BL(F_1) \ltbool \BL(F_2)$. The
converse does not hold, however. Indeed, from
Proposition\nobreakspace \ref {prop:converse-elimination}, it follows that,
e.g., $\BL(\conv{}) \ltbool \BL(\pi)$. From
Theorem\nobreakspace \ref {th:path-inclusion}, however, we know that $\BL(\conv{})
\not\ltpath \BL(\pi)$.


To accommodate the collapse of $\conv{}$ in our characterization of
$\ltbool$, we introduce some new notation. For a set of nonbasic
features $F$, define $\zlang{F}$ as follows.
\[\zlang{F} = 
\begin{cases} (F\setminus \{\conv{}\})\cup \{\pi\}, & \text{if ${}^{-1} \in \alang{F},\cap \not\in \alang{F}$} \\
              F, & \text{otherwise}
\end{cases}
\]
For example, $\zlang{\{\di, \conv{}\}} = \{ \di,\pi\}$.

We will establish the following characterization.
\begin{theorem} \label{th:bool-inclusion} 
  Let $F_1$ and $F_2$ be sets of nonbasic features. Then, $\elang{F_1} \boolq \elang{F_2}$ if and only if $F_1 \subseteq \alang{F_2}$ or $\zlang{F_1}\subseteq \alang{F_2}$ 
\end{theorem}

The ``if'' direction of Theorem\nobreakspace \ref {th:bool-inclusion} is shown by 
Proposition\nobreakspace \ref {th:path-if} (since $\ltpath$ implies $\ltbool$) and Proposition\nobreakspace  \ref {th:bool-if-cross}.

\begin{proposition}\label{th:bool-if-cross}
If $\zlang{F_1}\subseteq \alang{F_2}$ then $\elang{F_1}\ltbool\elang{F_2}$.
\begin{proof}
  We distinguish two cases. If $F_1\subseteq \alang{F_2}$, 
  then $\elang{F_1}\ltpath \elang{F_2}$, by Proposition\nobreakspace \ref {th:path-if}, 
  whence $\elang{F_1}\ltbool\elang{F_2}$. 

  In the other case, ${}^{-1} \in \alang{F_1}$, and $\cap \not\in \alang{F_1}$. 
  Hence $\elang{F_1} \ltbool \elang{F_1 \setminus \{\conv{}\} \cup \{\pi\}} = \elang{\zlang{F_1}}$ 
  by Proposition\nobreakspace \ref {prop:converse-elimination}. Furthermore, $\elang{\zlang{F_1}}\ltpath \elang{F_2}$ 
  since $\zlang{F_1}\subseteq \alang{F_2}$ by Proposition\nobreakspace \ref {th:path-if}, whence 
  $\elang{\zlang{F_1}}\ltbool\elang{F_2}$. Now, by transitivity $\elang{F_1}\ltbool\elang{F_2}$ as desired.
\end{proof}
\end{proposition}
The converse of this proposition does not hold in general, e.g., $\elang{\conv{}} \ltbool \elang{{}^{-1},\setminus}$ but $\zlang{\{{}^{-1}\}} = \{\pi\} \nsubseteq \alang{\{{}^{-1},\setminus\}} = \{{}^{-1},\setminus,\cap\}$. 

The ``only if'' direction of Theorem\nobreakspace \ref {th:bool-inclusion}, requires 
a detailed analysis, which proceeds along the same lines as the
analysis in Section\nobreakspace \ref {sec:path-queries}. We first establish the
``only if'' direction for the cases where $\BL(F_1)$ and $\BL(F_2)$
belong to the same class among $\clbot$ and $\clcap$, and then
consider the case where $\BL(F_1)$ and $\BL(F_2)$ belong to distinct
classes.

\subsection{Languages without $\cap$}
\label{sec:languages-bottom-bool}
In this subsection, we show the ``only if'' direction of
Theorem\nobreakspace \ref {th:bool-inclusion}, restricted to $\clbot$, the class of
languages without $\cap$.

\begin{proposition} \label{prop:bool-bottom-Hasse} Let $\BL(F_1)$ and
  $\BL(F_2)$ be in $\clbot$.  If $F_1 \not \subseteq \alang{F_2}$ and $\zlang{F_1}\nsubseteq \alang{F_2}$, then
  $\BL(F_1) \not\ltbool \BL(F_2)$.
\end{proposition}

Propositions\nobreakspace \ref{th:path-if},\nobreakspace\ref{th:bool-if-cross} and\nobreakspace  \ref {prop:bool-bottom-Hasse}
combined yield the Hasse diagram of $\ltbool$ for $\clbot$, shown in
Figure\nobreakspace \ref {fig:hasse-noint-bool}.  It is indeed readily verified that
for any two languages $\BL(F_1)$ and $\BL(F_2)$ in $\clbot$, there is
a path from $\BL(F_1)$ to $\BL(F_2)$ in
Figure\nobreakspace \ref {fig:hasse-noint-bool} if and only if $F_1 \subseteq
\alang{F_2}$ or $\zlang{F_1}\subseteq \alang{F_2}$.

Towards a proof of Proposition\nobreakspace \ref {prop:bool-bottom-Hasse}, we first establish the following.

\begin{proposition} \label{prop:conv-di} 
Let $F$ be a set of nonbasic features. If $\conv{} \in \fbar$, then
we have $\BL(F) \nltstrong \BL(\div)$.
  \begin{proof}
Let $\mathcal{C}$ be the class of all graphs $G$ such that $G$ is acyclic and $\adom(G)$ contains at least three elements, and let $e \in \elang{\di}$. We will show that on the class $\mathcal{C}$, the boolean query $e \neq \emptyset$ is either $\emptyset\neq \emptyset$ (always false) or equivalent to $R^m \neq \emptyset$ for some natural number $m$. Let us first show this for union-free expressions. Since $\di^i = \id\cup \di$ in $\mathcal{C}$ for $i > 1$, we may assume that $e = R^{n_1} \circ \di \circ R^{n_2} \circ \di \circ \ldots \circ \di \circ R^{n_k}$ where $k > 1$ and $n_1,\ldots,n_k$ are natural numbers greater than zero. We set $m$ to be the maximum of the $n_l$ for $1 \leq l \leq k$. Let $G$ be an arbitrary graph in $\mathcal{C}$. Clearly, if $e(G) \neq \emptyset$ then $R^{m}(G) \neq \emptyset$. For the other direction, assume $(x,y)\in R^{m}(G)$. Since $G$ is acyclic, $x\neq y$, so $(y,x)\in \di(G)$. Hence $(x,x)\in R^m\circ \di(G)$. For any $l\leq m$, we also have $(x,x)\in R^l\circ \di(G)$. We conclude that $(x,x)\in R^{n_1} \circ \di \circ \ldots\circ R^{n_k}\circ \di(G)$. In particular, $e(G)$ is nonempty as desired. 

For the claim to hold with union, it suffices to show it for a union of two union-free expressions. Indeed, the form $R^m$ is union-free! So, consider an expression $e$ of the form $R^{m_1} \cup R^{m_2}$. Then $e \neq \emptyset$ is equivalent to $R^{\min(m_1,m_2)}\neq \emptyset$, which proves the claim.

Now consider graphs $G_1$ and $G_2$ in Figure~\ref{fig:separations}(a). These graphs belong to $\mathcal{C}$, and are clearly indistinguishable by any expression of the form $R^{m}\neq \emptyset$. The graphs, however, are distinguishable by the boolean query $R^{2}\circ R^{-1}\circ R^2\neq \emptyset$.
\end{proof}
\end{proposition}

\begin{figure}[htb]
\begin{center}
 \resizebox{0.20\textwidth}{!}{
\begin{picture}(0,0)%
\includegraphics{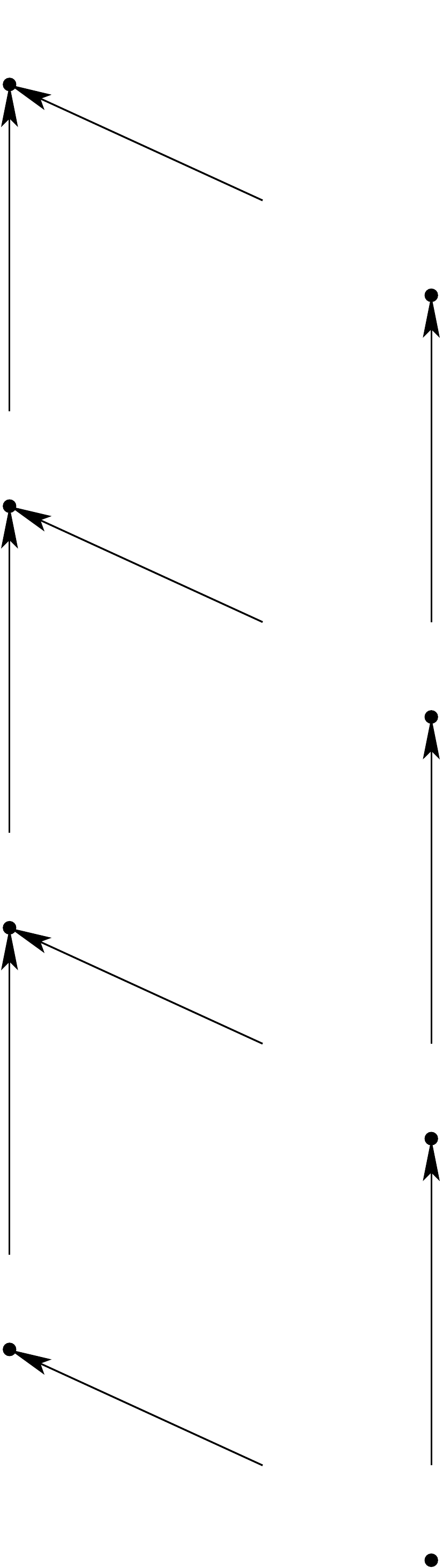}%
\end{picture}%
\setlength{\unitlength}{4144sp}%
\begingroup\makeatletter\ifx\SetFigFont\undefined%
\gdef\SetFigFont#1#2#3#4#5{%
  \reset@font\fontsize{#1}{#2pt}%
  \fontfamily{#3}\fontseries{#4}\fontshape{#5}%
  \selectfont}%
\fi\endgroup%
\begin{picture}(3764,13372)(7119,-10022)
\put(7201,-601){\makebox(0,0)[b]{\smash{{\SetFigFont{25}{30.0}{\familydefault}{\mddefault}{\updefault}{\color[rgb]{0,0,0}$\BL(\,\framebox[1.2\width]{\hhh${}^{-1},\textit{di},\pi$}\,)\makebox[0pt][l]{$\mathstrut=\BL(\,\framebox[1.3\width]{\hhh$\textit{di},\pi$}\,)$}$}%
}}}}
\put(10801,1199){\makebox(0,0)[b]{\smash{{\SetFigFont{25}{30.0}{\familydefault}{\mddefault}{\updefault}{\color[rgb]{0,0,0}$\BL(\,\framebox[1.3\width]{\hhh${}^{-1},\overline\pi$}\,,\pi)\makebox[0pt][l]{$\mathstrut=\BL(\,\framebox[1.7\width]{\hhh$\overline\pi$}\,,\pi)$}$}%
}}}}
\put(7201,2999){\makebox(0,0)[b]{\smash{{\SetFigFont{25}{30.0}{\familydefault}{\mddefault}{\updefault}{\color[rgb]{0,0,0}$\BL(\,\framebox[1.2\width]{\hhh${}^{-1},\textit{di},\overline\pi$}\,,\pi)\makebox[0pt][l]{$\mathstrut=\BL(\,\framebox[1.3\width]{\hhh$\textit{di},\overline\pi$}\,,\pi)$}$}%
}}}}
\put(10801,-2401){\makebox(0,0)[b]{\smash{{\SetFigFont{25}{30.0}{\familydefault}{\mddefault}{\updefault}{\color[rgb]{0,0,0}$\BL(\,\framebox[1.3\width]{\hhh${}^{-1},\pi$}\,)\makebox[0pt][l]{$\mathstrut=\BL(\,\framebox[1.7\width]{\hhh$\pi$}\,)$}$}%
}}}}
\put(7201,-7801){\makebox(0,0)[b]{\smash{{\SetFigFont{25}{30.0}{\familydefault}{\mddefault}{\updefault}{\color[rgb]{0,0,0}$\BL(\,\framebox[1.7\width]{\hhh$\textit{di}$}\,)$}%
}}}}
\put(7201,-4201){\makebox(0,0)[b]{\smash{{\SetFigFont{25}{30.0}{\familydefault}{\mddefault}{\updefault}{\color[rgb]{0,0,0}$\BL(\,\framebox[1.3\width]{\hhh${}^{-1},\textit{di}$}\,)$}%
}}}}
\put(10801,-9601){\makebox(0,0)[b]{\smash{{\SetFigFont{25}{30.0}{\familydefault}{\mddefault}{\updefault}{\color[rgb]{0,0,0}$\BL$}%
}}}}
\put(10801,-6001){\makebox(0,0)[b]{\smash{{\SetFigFont{25}{30.0}{\familydefault}{\mddefault}{\updefault}{\color[rgb]{0,0,0}$\BL(\,\framebox[1.7\width]{\hhh${}^{-1}$}\,)$}%
}}}}
\end{picture}%
}
\caption{The Hasse diagram of $\ltbool$ for $\clbot$. For each
  language, the boxed features are a minimal set of nonbasic features
  defining the language, while the other features can be derived from
  them in the sense of Theorem\nobreakspace \ref {th:path-inclusion} (using the
  appropriate interdependencies).}
\label{fig:hasse-noint-bool}
\end{center}
\end{figure}

As detailed below, Propositions\nobreakspace \ref {prop:path-bottom} and\nobreakspace  \ref {prop:conv-di} are now subsequently used to show that for every
pair $F_1$ and $F_2$ of sets of nonbasic features for which $F_1
\not\subseteq \ftilde_2$, that $\BL(F_1) \not\ltbool \BL(F_2)$, in the
same way as in Section\nobreakspace \ref {sec:languages-bottom} and
\ref{sec:languages-cap}.

The remainder of the proof of Proposition\nobreakspace \ref {prop:bool-bottom-Hasse}
is again a combinatorial analysis to verify that the above-mentioned
propositions cover all relevant cases. 

\begin{proof}[Proof of Proposition\nobreakspace \ref {prop:bool-bottom-Hasse}.]
      First, note that $F_1 \cup F_2 \subseteq \{{}^{-1},\cpi,\pi,\di\}$ since $\elang{F_1}$ and $\elang{F_2}$ are in $\mathcal{C}$. We will consider two cases: $\pi \in \alang{F_2}$ and $\pi \not\in \alang{F_2}$. First we will consider $\pi \in \alang{F_2}$. Since $\zlang{F_1}\nsubseteq \alang{F_2}$, there must be another feature, not equal to ${}^{-1}$ or $\pi$ present. If this feature is $\di$, then Proposition\nobreakspace \ref {prop:path-bottom}(\ref{bottom-di}) proves the result. On the other hand, if this feature is $\cpi$ then Proposition\nobreakspace \ref {prop:path-bottom}(\ref{bottom-cpi}) proves the result.
      
      Now consider the case where $\pi \not \in \alang{F_2}$. Here, $\cpi \not\in F_2$ and thus $F_1 \not\subseteq \alang{F_2}\subseteq \{\di,\conv{}\}$. Hence one of $\conv{},\cpi,\pi$ or $\di$ is present in $F_1$ but missing in $\alang{F_2}$. If that feature is $\conv{}$, then $\alang{F_2}\subseteq \{\di\}$, and hence Proposition\nobreakspace \ref {prop:conv-di} proves the result. On the other hand, if that feature is $\cpi,\di$ or $\pi$, the result follows directly from Proposition\nobreakspace \ref {prop:path-bottom}.
  \end{proof}


\subsection{Languages with $\cap$}
\label{sec:languages-cap-bool}
In this subsection, we show the ``only if'' direction of
Theorem\nobreakspace \ref {th:bool-inclusion}, restricted to $\clcap$, the class of
languages with $\cap$.

\begin{proposition} \label{prop:bool-cap-Hasse} Let $\BL(F_1)$ and
  $\BL(F_2)$ be in $\clcap$.  If $F_1 \not \subseteq \alang{F_2}$ and $\zlang{F_1}\nsubseteq \alang{F_2}$, then 
  $\BL(F_1) \nltstrong \BL(F_2)$.
\end{proposition}

Notice that since $\cap \in \alang{F_1}$, $\zlang{F_1}= F_1$. Hence,
Theorem\nobreakspace \ref {th:path-inclusion} and
Proposition\nobreakspace \ref {prop:bool-cap-Hasse} combined show that $\ltbool$
coincides with $\ltpath$ on $\clcap$.  As a result, the Hasse diagram
of $\ltbool$ for $\clcap$ is the same as the Hasse diagram of
$\ltpath$ for $\clcap$ shown in Figure\nobreakspace \ref {fig:hasse-int}. Note that,
in addition, all separations are strong.

Towards a proof of Proposition\nobreakspace \ref {prop:bool-cap-Hasse}, we first
establish the following.
\begin{proposition} \label{prop:bool-converse-cap}
  Let $F_1$ and $F_2$ be sets of nonbasic features. If $\conv{} \in
  F_1$, $\cap \in F_1$, and $\conv{} \not\in F_2$, then $\BL(F_1)
  \nltstrong \BL(F_2)$.
\end{proposition}

\begin{proof}
  The graphs $G_1$ and $G_2$ shown in
  Figure\nobreakspace \ref {fig:separations}~(c), top and bottom, are distinguished
  by the boolean query $q$ expressed by $(R^2 \circ \conv{R} \circ R)
  \cap R$. On these graphs, the Brute-Force Algorithm of
  Section\nobreakspace \ref {sec:brute-force-approach} does not terminate in a
  reasonable time. It can be verified in polynomial
  time, however, that for each pair $(a_1,b_1) \in \adom(G_1)^2$,
  there exists $(a_2,b_2) \in \adom(G_2)^2$ such that $(G_1, a_1, b_1)
  \simeq_k (G_2, a_2, b_2)$ for any depth $k$ \cite{rabisim}.
From
  Proposition\nobreakspace \ref {prop-bisimilar-noparameter}, it follows that $q$ is
  not expressible in $\BL(F_2)$.
\end{proof}

The remainder of the proof of Proposition\nobreakspace \ref {prop:bool-cap-Hasse}
proceeds as the proof of Proposition\nobreakspace \ref {prop:path-int-Hasse}, except
that Proposition\nobreakspace \ref {prop:bool-converse-cap} is used instead of
Proposition\nobreakspace \ref {prop:path-bottom}~(\ref{bottom-conv}).

\subsection{Cross-relationships between subdiagrams}
\label{sec:cross-rel-bool}
To finish the proof of Theorem\nobreakspace \ref {th:bool-inclusion}, we finally
show the ``only if'' direction for the case where $\BL(F_1)$ and
$\BL(F_2)$ belong to different classes.

\begin{proposition} \label{prop:bool-cross-Hasse} Let $\BL(F_1)$ and
  $\BL(F_2)$ be languages such that one language belongs to $\clbot$,
  and the other language belongs to $\clcap$. If $F_1 \not \subseteq \alang{F_2}$ and $\zlang{F_1}\nsubseteq \alang{F_2}$, then $\BL(F_1) \not \ltbool \BL(F_2)$.
\end{proposition}

Towards a proof of Proposition\nobreakspace \ref {prop:bool-cross-Hasse}, we first
establish the following.

\begin{proposition} \label{prop:conv-cross} Let $F_1$ be a set of
  nonbasic features. If $\conv{} \in \fbar_1$, and $F_2 \subseteq \{
  \setminus, \cap\}$, then $\BL(F_1) \nltstrong \BL(F_2)$.
\end{proposition}

\begin{proof}
  Consider the graphs $G_1$ and $G_2$ displayed in Figure\nobreakspace \ref {fig:separations}~(a) and define $R^0(G_i)$ to equal $\id(G_i)$ for $i = 1,2$. First, notice that $\id(G_i)$, $R(G_i)$,and $R^2(G_i)$ are pairwise disjoint for $i = 1,2$. Utilizing this, it can be proven by straightforward induction that for every $e \in \elang{\setminus}$ there exists $Z \subseteq \{0,1,2\}$ such that $e(G_1) = \cup_{i \in Z} R^i(G_1)$ and $e(G_2) = \cup_{i \in Z} R^i(G_2)$. This clearly implies that $G_1$ and $G_2$ are indistinguishable in $\elang{\setminus}$. whence they are also indistinguishable in $\elang{F_2}$ as well since $\elang{F_2}\ltbool \elang{\setminus}$. The graphs, however, are distinguishable
  by the boolean query expressed by $R^2 \circ \conv{R} \circ R^2$.
\end{proof}

As detailed below, Propositions\nobreakspace \ref {prop:path-bottom},  \ref {prop:path-int},  \ref {prop:path-cross-comp} and\nobreakspace  \ref {prop:conv-cross}
are now subsequently used to show that for every pair $F_1$ and $F_2$
of sets of nonbasic features for which $F_1 \not \subseteq \alang{F_2}$ and $\zlang{F_1}\nsubseteq \alang{F_2}$,
that $\BL(F_1) \not\ltbool \BL(F_2)$, in the same way as in
Sections\nobreakspace \ref {sec:languages-bottom-bool} and\nobreakspace  \ref {sec:languages-cap-bool}.

The remainder of the proof of Proposition\nobreakspace \ref {prop:bool-cross-Hasse}
is again a combinatorial analysis to verify that the above-mentioned
propositions cover all relevant cases.  

\begin{proof}[Proof of Proposition\nobreakspace \ref {prop:bool-cross-Hasse}]
  If $F_1 \in \clcap$ and $F_2 \in \clbot$, then $\cap \in \alang{F_1}$ and 
  $\cap\not\in \alang{F_2}$. Hence Proposition\nobreakspace \ref {prop:path-cross-comp} directly implies our result. 

  Conversely, if $F_1 \in \clbot$ and $F_2\in \clcap$, then $x \in \{\di,\pi,\cpi,\conv{}\}$ is present in $F_1$, but lacking in $\alang{F_2}$. We will now consider every such $x$.

   If $x \in \{\di,\pi,\cpi\}$ then the proof proceeds as the proof of Proposition\nobreakspace \ref {prop:path-int-Hasse}.
   
  If $x = \conv{}$, then $\zlang{F_1} = (F_1\setminus \{\conv{}\}) \cup \{\pi\}$ since $F_1 \in \clbot$. 
  Furthermore, by hypothesis, there is a feature $x$ present in $\zlang{F_1}$ 
  which is not present in $\alang{F_2}$. Notice that $x \neq \conv{}$. 
  If $x \neq \pi$, then there exists a feature in $F_1$ other than $\conv{}$ 
  which is missing in $\alang{F_2}$, hence the result follows from the 
  previous case. On the other hand, if $x = \pi$, then 
  $F_2 \cap \{ \di,\pi,\cpi,\conv{} \} = \emptyset$. Hence 
  $F_2 \subseteq \{\setminus,\cap\}$, and thus the result 
  follows directly from Proposition\nobreakspace \ref {prop:conv-cross}.

\end{proof}

Propositions\nobreakspace  \ref{th:path-if},\nobreakspace\ref{th:bool-if-cross},\nobreakspace  \ref {prop:bool-bottom-Hasse}, \ref {prop:bool-cap-Hasse} and\nobreakspace  \ref {prop:bool-cross-Hasse}, together
prove Theorem\nobreakspace \ref {th:bool-inclusion}.

Hence, the Hasse diagram of $\ltbool$ can be obtained from the
subdiagrams for $\clbot$, and $\clcap$ by simply adding arrows from $\BL$ to $\elang{\cap}$, $\elang{\di,\pi}$ to $\elang{\cap,\di,\pi}$, $\elang{\pi}$ to $\elang{\cap,\pi}$, $\elang{\cpi,\pi}$ to $\elang{\cap,\cpi,\pi}$ and $\elang{\di,\cpi,\pi}$ to $\elang{\cap,\di,\cpi,\pi}$. 
So, all paths between the subdiagrams are induced by these arrows, the 5 equations from the beginning of
Section\nobreakspace \ref {sec:path-queries}, and
Proposition\nobreakspace \ref {prop:converse-elimination}.

\section{Further research}
\label{sec:further-research}
There are alternative modalities for expressing boolean queries apart
from interpreting the nonemptiness of an expression as the value true and
emptiness as the value false.  For example,
one possibility is to consider a boolean query
$q$ expressible if there are two expressions $e_1$ and $e_2$ such that
$e_1(G) \subseteq e_2(G)$ if, and only if, $q(G)$ is true, for all $G$.
For some of our languages, such alternative modalities would not make
a difference, but it would for others. Looking into these alternative
modalities is an interesting topic for further research.

In the present paper, we have been focusing on expressive power,
but, of course, it is also interesting to investigate the
decidability of satisfiability or containment of expressions.
Much is already known.  From the undecidability of FO$^3$, it
follows that the most powerful language is undecidable, and the
same holds even without converse.  From the decidability of ICPDL
\cite{icpdl_jsl}, all languages without set difference have a
decidable satisfiability problem, although this is not yet known
for satisfiability restricted to finite relations.  An
interesting question is the decidability of satisfiability or
validity of the languages with set difference, but without the
diversity relation.  Recently, it has been shown that finite
satisfiability for the quite weak fragment $\BL(\setminus)$
without $\id$, formed by the operators union, composition, set
difference and nothing else, over a single binary relation, is
still undecidable \cite{tony_da_arxiv}.

Another natural question is whether the notion of arrow logic
bisimulation, that we use as a tool to prove some nonexpressibility
results, can actually be adapted to obtain characterizations of
indistinguishability in the various languages, as is the
case for modal logic \cite{gorankotto}.  We have in fact done this
for all languages with intersection
\cite{rabisim}.  A further question then is whether van
Benthem-style expressive completeness results \cite{otto_durham} can be
established.

Finally, there are still other interesting operators on binary
relations that can be considered.  A good example is
residuation \cite{pratt_relcalc}, a derived operator of the calculus of
relations, and interesting to consider separately, as we have done for
projection and coprojection. Residuation is interesting from a
database perspective because it corresponds to the set containment
join \cite{mamoulis_setjoin}.

\section*{Acknowledgment}

We thank the anonymous referees for their constructive feedback.
We thank Balder ten Cate and Maarten Marx for helpful information
on the question of succinctness of FO$^3$ compared to the algebra
$\mathcal N(\div,\conv{},-)$.

\newcommand{\etalchar}[1]{$^{#1}$}

\end{document}